\documentclass[final, conference, fontsize=10pt]{IEEEtran}
\ifCLASSINFOpdf
\else
\fi

\usepackage{graphicx,graphics,epsfig,epstopdf,float}
\usepackage{amsmath} \usepackage{amssymb}  \usepackage{mathrsfs}
\usepackage{enumerate}
\usepackage{subfigure}
\usepackage{booktabs}
\usepackage{colortbl}
\usepackage{color}  \usepackage{bm}
\usepackage{psfrag}
\usepackage{cite}
\usepackage{algorithm}
\usepackage{algorithmicx}
\usepackage{algpseudocode}
\usepackage{mdwlist}
\usepackage{longtable}
\usepackage{array}
\usepackage{acronym}  \usepackage{svg}
\usepackage{bm}
\usepackage{url}
\usepackage{xcolor}
\usepackage{stfloats}
\usepackage{cases}
\def\BibTeX{{\rm B\kern-.05em{\sc i\kern-.025em b}\kern-.08em
    T\kern-.1667em\lower.7ex\hbox{E}\kern-.125emX}}

\DeclareMathAlphabet{\mathsfbr}{OT1}{cmss}{m}{n}\SetMathAlphabet{\mathsfbr}{bold}{OT1}{cmss}{bx}{n}\DeclareRobustCommand{\msf}[1]{\ifcat\noexpand#1\relax\msfgreek{#1}\else\mathsfbr{#1}\fi}

\makeatletter
\newcommand{\msfgreek}[1]{\csname s\expandafter\@gobble\string#1\endcsname}
\makeatother

\DeclareFontEncoding{LGR}{}{} \DeclareSymbolFont{sfgreek}{LGR}{cmss}{m}{n}
\SetSymbolFont{sfgreek}{bold}{LGR}{cmss}{bx}{n}
\DeclareMathSymbol{\salpha}{\mathord}{sfgreek}{`a}
\DeclareMathSymbol{\sbeta}{\mathord}{sfgreek}{`b}
\DeclareMathSymbol{\sgamma}{\mathord}{sfgreek}{`g}
\DeclareMathSymbol{\sdelta}{\mathord}{sfgreek}{`d}
\DeclareMathSymbol{\sepsilon}{\mathord}{sfgreek}{`e}
\DeclareMathSymbol{\szeta}{\mathord}{sfgreek}{`z}
\DeclareMathSymbol{\seta}{\mathord}{sfgreek}{`h}
\DeclareMathSymbol{\stheta}{\mathord}{sfgreek}{`j}
\DeclareMathSymbol{\siota}{\mathord}{sfgreek}{`i}
\DeclareMathSymbol{\skappa}{\mathord}{sfgreek}{`k}
\DeclareMathSymbol{\slambda}{\mathord}{sfgreek}{`l}
\DeclareMathSymbol{\smu}{\mathord}{sfgreek}{`m}
\DeclareMathSymbol{\snu}{\mathord}{sfgreek}{`n}
\DeclareMathSymbol{\sxi}{\mathord}{sfgreek}{`x}
\DeclareMathSymbol{\somicron}{\mathord}{sfgreek}{`o}
\DeclareMathSymbol{\spi}{\mathord}{sfgreek}{`p}
\DeclareMathSymbol{\srho}{\mathord}{sfgreek}{`r}
\DeclareMathSymbol{\ssigma}{\mathord}{sfgreek}{`s}
\DeclareMathSymbol{\stau}{\mathord}{sfgreek}{`t}
\DeclareMathSymbol{\supsilon}{\mathord}{sfgreek}{`u}
\DeclareMathSymbol{\sphi}{\mathord}{sfgreek}{`f}
\DeclareMathSymbol{\schi}{\mathord}{sfgreek}{`q}
\DeclareMathSymbol{\spsi}{\mathord}{sfgreek}{`y}
\DeclareMathSymbol{\somega}{\mathord}{sfgreek}{`w}

\DeclareMathSymbol{\svarsigma}{\mathord}{sfgreek}{`c}

\DeclareMathSymbol{\sGamma}{\mathalpha}{sfgreek}{`G}
\DeclareMathSymbol{\sDelta}{\mathalpha}{sfgreek}{`D}
\DeclareMathSymbol{\sTheta}{\mathalpha}{sfgreek}{`J}
\DeclareMathSymbol{\sLambda}{\mathalpha}{sfgreek}{`L}
\DeclareMathSymbol{\sXi}{\mathalpha}{sfgreek}{`X}
\DeclareMathSymbol{\sPi}{\mathalpha}{sfgreek}{`P}
\DeclareMathSymbol{\sSigma}{\mathalpha}{sfgreek}{`S}
\DeclareMathSymbol{\sUpsilon}{\mathalpha}{sfgreek}{`U}
\DeclareMathSymbol{\sPhi}{\mathalpha}{sfgreek}{`F}
\DeclareMathSymbol{\sPsi}{\mathalpha}{sfgreek}{`Y}
\DeclareMathSymbol{\sOmega}{\mathalpha}{sfgreek}{`W}

\DeclareRobustCommand{\mcal}[1]{\ifcat\noexpand#1\relax\mathnormal{#1}\else\cal{#1}\fi
}
\DeclareRobustCommand{\BM}[1]{\ifcat\noexpand#1\relax\bm{\boldUppercaseItalicGreek{#1}}\else\bm{#1}\fi
}
\makeatletter
\newcommand{\boldUppercaseItalicGreek}[1]{\csname var\expandafter\@gobble\string#1\endcsname}
\makeatother
\newcommand{\rv}[1]{\MakeLowercase{\msf{#1}}} \newcommand{\RV}[1]{\bm{\MakeLowercase{\msf{#1}}}}  \newcommand{\RM}[1]{\bm{\MakeUppercase{\msf{#1}}}}   

\newcommand{\V}[1]{\bm{#1}} \newcommand{\M}[1]{\BM{#1}}  \newcommand{\zz}[1]{\mathrm{#1}} \newcommand{\ermtrans}{^\text{H}}
\newcommand{\trans}{^\text{T}}

\definecolor{BLUE}{rgb}{0,0,1}
\newtheorem{lemma}{Lemma}
\newtheorem{theorem}{Theorem}
\newtheorem{proposition}{Proposition}

\newtheorem{definition}{Definition}
\newtheorem{remark}{Remark}

\newenvironment{proof}[1][Proof]{{\indent \it #1:\quad}}{\hfill \IEEEQED\par}

\renewcommand{\IEEEQED}{\IEEEQEDopen} \acrodef{gnss}[GNSS]{global navigation satellite system}\acrodef{rf}[RF]{radio frequency}
\acrodef{aoa}[AOA]{angle-of-arrival}
\acrodef{rss}[RSS]{received signal strength}
\acrodef{toa}[TOA]{time-of-arrival}
\acrodef{tdoa}[TDOA]{time-difference-of-arrival}
\acrodef{rtt}[RTT]{round-trip time}
\acrodef{fdd}[FDD]{frequency division duplex}
\acrodef{tdd}[TDD]{time division duplex}
\acrodef{fd}[FD]{full-duplex}
\acrodef{sdp}[SDP]{semidefinite programming}
\acrodef{crlb}[CRLB]{Cram\'{e}r-Rao lower bound}
\acrodef{nc}[NC]{narrow correlator}
\acrodef{sc}[SC]{strobe correlator}
\acrodef{pll}[PLL]{phase locked loop}
\acrodef{mp}[MP]{multipath}
\acrodef{sp}[SP]{single path}
\acrodef{ff}[FF]{flat fading}
\acrodef{mds}[MDS]{multidimensional scaling}
\acrodef{snr}[SNR]{signal-to-noise ratio}
\acrodef{los}[LOS]{line-of-sight}
\acrodef{nlos}[NLOS]{non-line-of-sight}
\acrodef{sic}[SIC]{serial interference cancelation}
\acrodef{pic}[PIC]{parallel interference cancelation}
\acrodef{adc}[ADC]{analog-to-digital converter}
\acrodef{bp}[BP]{basis pursuit}
\acrodef{lasso}[LASSO]{least absolute shrinkage and selection operator}
\acrodef{omp}[OMP]{orthogonal matching pursuit}
\acrodef{lls}[LLS]{linear least squares}
\acrodef{wlls}[WLLS]{weighted linear least squares}
\acrodef{nlls}[NLLS]{nonlinear least squares}
\acrodef{awgn}[AWGN]{additive white Gaussian noise}
\acrodef{cirf}[CIRF]{channel impulse response function}
\acrodef{irf}[IRF]{impulse response function}
\acrodef{llr}[LLR]{log-likelihood ratio}
\acrodef{llrs}[LLRs]{log-likelihood ratios}
\acrodef{fim}[FIM]{Fisher information matrix}
\acrodef{efim}[EFIM]{equivalent Fisher information matrix}
\acrodef{mse}[MSE]{mean squared error}
\acrodef{peb}[PEB]{position error bound}
\acrodef{rmse}[RMSE]{root mean squared error}
\acrodef{seb}[SEB]{synchronization error bound}
\acrodef{imu}[IMU]{inertial measurement unit}
\acrodef{gm}[GM]{Gaussian-mixture}
\acrodef{reb}[REB]{ranging error bound}
\acrodef{co}[CO]{clock offset}
\acrodef{pdf}[PDF]{probability density function}\acrodef{iot}[IoT]{internet-of-things}
\acrodef{los}[LOS]{line-of-sight}
\acrodef{nlos}[NLOS]{non-line-of-sight}
\acrodef{phd}[PHD]{probability hypothesis density}
\acrodef{slam}[SLAM]{simultaneous localization and mapping}
\acrodef{mpc}[MPC]{multipath component}
\acrodef{toa}[ToA]{time-of-arrival}
\acrodef{aoa}[AoA]{angle-of-arrival}
\acrodef{vt}[VT]{virtual transmitter}
\acrodef{rbpf}[RBPF]{Rao-Blackwellised particle filter}
\acrodef{GM}{Gaussian-mixture}
\acrodef{uwb}[UWB]{ultra-wideband}
\acrodef{rfs}[RFS]{random finite set}
\acrodef{cir}[CIR]{channel impulse response}
\acrodef{ekf}[EKF]{extended Kalman filter}
\acrodef{fisst}[FISST]{finite-set statistics}
\acrodef{map}[MAP]{maximum-a-posteriori}
\acrodef{cdf}[CDF]{cumulative distribution function}
\acrodef{mmse}[MMSE]{minimum-mean-square-error}
\acrodef{xr}[XR]{extended reality}
\acrodef{gc}[GC]{Gaussian component}
\acrodef{ris}[RIS]{reconfigurable intelligent surface}
\acrodef{rhs}[RHS]{right-hand side}
\acrodef{fov}[FoV]{field of view}
\acrodef{da}[DA]{data association}
\acrodef{smc}[SMC]{sequential Monte Carlo}
\acrodef{eap}[EAP]{expected-a-posteriori}
\acrodef{ppp}[PPP]{Poisson point process}
\acrodef{isac}[ISAC]{integrated sensing and communication}
\acrodef{mimo}[MIMO]{multiple-input and multiple-output}
\acrodef{siso}[SISO]{single-input and single-output}
\acrodef{csi}[CSI]{channel state information}
\acrodef{iid}[i.i.d.]{independently and identically distributed}
\acrodef{mcmc}[MCMC]{Markov chain Monte Carlo}
\acrodef{itur}[ITU-R]{radio communication division of the international telecommunications union}
\acrodef{hrllc}[HRLLC]{hyper-reliable and low-latency communications}
\acrodef{v2x}[V2X]{vehicular-to-everything}
\acrodef{dmc}[DMC]{discrete memoryless channel}
\acrodef{crb}[CRB]{Cram\'{e}r-Rao bound}
\acrodef{sac}[S\&C]{sensing and communication}
\acrodef{miso}[MISO]{multiple-input and single-output}
\acrodef{bs}[BS]{base station}
\acrodef{ue}[UE]{user equipment}
\acrodef{bcrb}[BCRB]{Bayesian Cram\'{e}r-Rao bound}
\acrodef{ecrb}[ECRB]{expected Cram\'{e}r-Rao bound}
\acrodef{bfim}[BFIM]{Bayesian Fisher information matrix}
\acrodef{dmt}[DMT]{diversity-multiplexing tradeoff}
\acrodef{svd}[SVD]{singular value decomposition}
\acrodef{evd}[EVD]{eigenvalue decomposition}
\acrodef{dof}[DoF]{Degree of Freedom}

\newcommand{\paperTitle}{
Sensing-Constrained Diversity-Multiplexing Tradeoff in MIMO ISAC: A Geometric Approach
}
\newcommand{\Complex}{\mathbb{C}}
\newcommand{\Real}{\mathbb{R}}

\newcommand{\Prob}{\mathbb{P}}
\newcommand{\Expect}{\mathbb{E}}

\newcommand{\hc}{\RM{H}_\zz{c}}
\newcommand{\hcnb}{\M{H}_\zz{c}}
\newcommand{\nc}{N_\zz{c}}
\newcommand{\snrc}{\eta_{\zz{c}}}
\newcommand{\snrs}{\eta_{\zz{s}}}
\newcommand{\rank}[1]{\zz{rk}(\M{#1})}
\newcommand{\indep}{\perp \!\!\!\! \perp}

\begin{document}
\setlength{\columnsep}{0.201 in}
\title{\paperTitle}
\author{
    Yinuo Du$^*$,
    Ziping Lu$^*$,
    Xiao Shen$^*$,
    Hanying Zhao$^*$,
    and Yuan Shen$^*$\\
    $^*$Department of Electronic Engineering, BNRist, Tsinghua University, Beijing, China\\
     Email: \{duyn24, lzp23, shenx20\}@mails.tsinghua.edu.cn, hanying@kth.se, shenyuan\_ee@tsinghua.edu.cn
}
\maketitle

\begin{abstract}
Diversity and multiplexing are the two fundamental gains of \ac{mimo} communications, enabling systems to simultaneously achieve increased reliability and higher data rates. The intricate interplay between these two metrics is captured by the celebrated \ac{dmt}. With the rapid evolution of wireless technologies, low-latency \ac{isac} has emerged as a key enabler for 6G applications, including \ac{xr} and massive digital twins. Consequently, understanding the \ac{dmt} within \ac{mimo} \ac{isac} systems becomes critical. In this paper, we investigate the communication \ac{dmt} in a mono-static \ac{mimo} \ac{isac} system under Rayleigh fading, specifically when the transmitter is constrained to emit sensing-optimal waveforms. By unveiling the geometric properties of generalized Stiefel manifolds and employing large-deviation analysis, we characterize the asymptotic outage probability of this typical \ac{isac} channel. This formulation yields an elegant converse bound on the sensing-constrained \ac{dmt}. Ultimately, our work provides a primary answer to a pivotal unanswered question in \ac{isac} system design: How much \ac{mimo} gain is fundamentally sacrificed in communication to integrate optimal sensing capabilities?
\end{abstract}

\acresetall

\acresetall	
\section{Introduction}\label{sec:intro}
Recognized as a pivotal technology for 6G \cite{ITU:R23}, \ac{isac} outperforms traditional separated designs by efficiently sharing hardware and spectral resources \cite{XioLiuCui:J23}. To fully realize its potential, \ac{mimo} technology is indispensable: it endows the system with high spatial resolution for sensing while unlocking critical diversity and multiplexing gains for communication \cite{GonFurKal:J24}. Consequently, \ac{mimo} \ac{isac} has emerged as a cornerstone for next-generation applications such as autonomous vehicular networks and digital twins.

Link-level \ac{isac} encompasses two primary paradigms: mono-static \cite{XioLiuCui:J23,AhmKobWig:J22}, where Tx performs sensing using known codewords, and bi-static \cite{SheLuZha:J25}, where Rx jointly senses and decodes. Due to Tx's perfect knowledge of the transmitted signals, the mono-static approach achieves superior sensing accuracy and is thus of greater practical interest \cite{GonFurKal:J24}. Theoretically, \ac{isac} is studied under either the infinite blocklength assumption \cite{XioLiuCui:J23,AhmKobWig:J22} using Shannon capacity, or the finite blocklength regime \cite{SheLuZha:J25} where decoding error probabilities are strictly non-zero. Since 6G targets ultra-low latency \cite{ITU:R23} and sensing is inherently time-sensitive, this paper specifically investigates a mono-static \ac{mimo} \ac{isac} channel within the finite blocklength regime.

In \ac{mimo} communications, the fundamental interplay between reliability and data rates is elegantly captured by the \ac{dmt} \cite{ZheTse:J03}. For \ac{mimo} \ac{isac} systems, the theoretical sensing-communication tradeoff has been recently established using the \ac{crb} and achievable rate \cite{XioLiuCui:J23}. However, this ergodic, zero-error formulation is incompatible with 6G applications such as \ac{hrllc} and neglects the diversity gain brought by \ac{mimo}. Consequently, to reflect the finite-blocklength realities of practical \ac{mimo} \ac{isac} systems, characterizing the sensing-communication tradeoff directly through the \ac{dmt} framework is of critical importance.

This paper investigates the communication \ac{dmt} of a mono-static \ac{mimo} \ac{isac} Rayleigh fading channel in the finite-blocklength regime. Under the strict constraint of emitting sensing-optimal waveforms, transmit codewords are geometrically confined to a generalized Stiefel manifold. Leveraging Riemannian geometry and large-deviation analysis, we characterize the asymptotic outage probability to derive an elegant converse bound on the sensing-constrained \ac{dmt}. Ultimately, we provide a preliminary answer to the pivotal unanswered question in \ac{mimo} \ac{isac}: How much \ac{mimo} communication gain must be sacrificed to guarantee optimal sensing?

\textit{Notations:} Random variables, vectors, and matrices are denoted by $\rv{x}$, $\RV{x}$, and $\RM{X}$, with their respective realizations given by $x$, $\V{x}$, and $\M{X}$. For a vector $\V{x}$, $\V{x}_{a:b}$ denotes the vector $\left(x_a,x_{a+1},...,x_{b}\right)$. The relation $a(x)\sim b(x)$ signifies that $\lim_{x\to \infty}a(x)/b(x)= 1$, the relation $a(x)\doteq b(x)$ signifies that $\log a(x) \sim \log b(x)$. We also use $(\cdot)^{+}\triangleq\max \{0,\cdot\}$.

\section{Problem Formulation}\label{sec:sysmod}
\subsection{\ac{mimo} \ac{isac} Model and Scheme}
We consider a mono-static \ac{mimo} \ac{isac} system where a \ac{bs} down-link communicates to a \ac{ue} while simultaneously performing radar sensing. The \ac{bs} is equipped with $M$ Tx antennas and $N_\zz{s}$ Rx antennas, whereas the \ac{ue} is equipped with $N_\zz{c}$ Rx antennas. Specifically, the \ac{sac} channels are given by:
\begin{subnumcases}{\label{eq:channelmodel}}
    \RM{Y}_\zz{s}=\sqrt{{\snrs}/{M}}\RM{H}_\zz{s}\RM{X}+\RM{Z}_\zz{s}, \label{eq:channelmodel_s}\\
    \RM{Y}_\zz{c}=\sqrt{{\snrc}/{M}}\RM{H}_\zz{c}\RM{X}+\RM{Z}_\zz{c} \label{eq:channelmodel_c}
\end{subnumcases}
where $\RM{Y}_\zz{s}$ and $\RM{Y}_\zz{c}$ denote the received \ac{sac} signals, and $\RM{X}\in\Complex^{M\times T}$ is the transmitted \ac{isac} codeword over a blocklength of $T\geq M$. 
$\RM{H}_\zz{s} = \M{G}(\RV{\eta}) \in \Complex^{N_\zz{s}\times M}$ represents the sensing channel response uniquely determined by a latent random parameter $\RV{\eta}\in\Real^{K}$ via an injective mapping $\M{G}(\cdot)$. The communication channel $\RM{H}_\zz{c}\in\Complex^{N_\zz{c}\times M}$ undergoes quasi-static Rayleigh fading and is assumed to be independent of $\RM{H}_\zz{s}$. The terms $\RM{Z}_\zz{s}$ and $\RM{Z}_\zz{c}$ represent additive white Gaussian noise. 
Without loss of generality, we assume $\RM{X},\RM{H}_\zz{s},\RM{H}_\zz{c}$ have normalized power, i.e. $\Expect\left(\text{tr}\left(\RM{X}\RM{X}\ermtrans\right)\right)/MT=\Expect\left(\text{tr}\left(\RM{H}_\zz{s}\RM{H}_\zz{s}\ermtrans\right)\right)/N_\zz{s}M=1$, and the entries of $\RM{H}_\zz{c}$, $\RM{Z}_\zz{s}$, and $\RM{Z}_\zz{c}$ are modeled as i.i.d. $\mathcal{CN}(0,1)$. Under these normalizations, $\snrs$ and $\snrc$ define the receive \acp{snr} at the sensing Rx and \ac{ue}, respectively.

For an \ac{isac} transmission, the \ac{bs} first selects a message $\rv{j}$ uniformly from the set $\mathcal{J}=\{1,2,...,J\}$ and maps $\rv{j}$ to a codeword $\RM{X} = \V{\mu}(\rv{j})$ from the codebook $\mathcal{X} \subset \Complex^{M\times T}$, where $\V{\mu}:\mathcal{J}\rightarrow\mathcal{X}$ is the bijective encoder. Finally, $\RM{X}$ is transmitted through the wireless channel \eqref{eq:channelmodel} to the \ac{sac} receivers.

At the communication receiver, we assume a coherent scheme where the channel matrix $\RM{H}_\zz{c}$ is known to the \ac{ue}. Upon receiving $\RM{Y}_\zz{c}$, the \ac{ue} estimates the transmitted message as $\hat{\rv{j}}=\V{\gamma}(\RM{Y}_\zz{c},\RM{H}_\zz{c})\in\mathcal{J}$, where $\V{\gamma}:\Complex^{N_\zz{c}\times T}\times\Complex^{N_\zz{c}\times M} \rightarrow\mathcal{J}$ denotes the communication decoding function.

At the sensing receiver, the latent parameter $\RV{\eta}$ is estimated using both the received signal $\RM{Y}_\zz{s}$ and the known transmitted codeword $\RM{X}$. This process yields the estimate $\hat{\RV{\eta}}=\omega(\RM{Y}_\zz{s},\RM{X})$, where $\omega:\Complex^{N_\zz{s}\times T}\times\mathcal{X}\rightarrow\Real^{K}$ defines the estimation function.

\begin{remark}
The sensing channel model \eqref{eq:channelmodel_s} has been used in the seminal work \cite{XioLiuCui:J23}, whereas we further assume that the communication channel \eqref{eq:channelmodel_c} falls under Rayleigh fading. 
As illustrated in Fig. \ref{fig:system}, such \ac{sac} channel models are of particular practical interest in complex urban scenarios such as low-altitude aerial networks, where target sensing is dominated by strong \ac{los} radar echoes while the communication link to the \ac{ue} suffers from severe scattered fading.
\end{remark}

\subsection{\ac{sac} Evaluation Metrics}
To evaluate sensing performance, we adopt the Miller-Chang-type \ac{bcrb}. Specifically, by defining the conditional \ac{bfim} as
\begin{align}
    \RM{J}_{\RV{\eta}|\RM{X}}&\triangleq\Expect_{\RM{Y}_\zz{s},\RV{\eta}|\RM{X}}\left[
    [\partial \ln p(\RM{Y}_\zz{s}|\RM{X},\RV{\eta})/{\partial \V{\eta}}]
    [\partial \ln p(\RM{Y}_\zz{s}|\RM{X},\RV{\eta})/{\partial \V{\eta}}]\trans
    \right]\nonumber
    \\
    &+\Expect_{\RV{\eta}}\left[
    [{\partial \ln p(\RV{\eta})}/{\partial \V{\eta}}]
    [{\partial \ln p(\RV{\eta})}/{\partial \V{\eta}}]\trans
    \right],\label{eq:bfim}
\end{align}
the corresponding Cram\'{e}r-Rao inequality is given by:
\begin{equation}
    \label{eq:MillerChang}
    \Expect\left(\|\hat{\RV{\eta}}-\RV{\eta}\|_2^2\right)\geq\Expect_{\RM{X}}\left[\text{tr}\left(\RM{J}_{\RV{\eta}|\RM{X}}^{-1}\right)\right]\triangleq e,
\end{equation}
where $e$ denotes the \ac{bcrb} for the \ac{mse} of the latent parameter $\RV{\eta}$.

For communication performance, there exist two fundamental metrics: the transmission rate $R\triangleq (\log J)/T$ and the decoding error probability $P_\zz{e}=\Prob(\hat{\rv{j}}\neq\rv{j})$. In an ergodic channel with infinite code length, $P_\zz{e}$ is irrelevant as Shannon capacity can be used. However, this obscures the critical \ac{mimo} diversity gains. Therefore, we evaluate the communication link through the rigorous lens of diversity and multiplexing gains \cite{ZheTse:J03}. Specifically, an \ac{isac} scheme achieves a multiplexing gain $r\geq 0$ and a diversity gain $d\geq 0$ if
\begin{equation}
    \lim_{\snrc\to\infty}R/\log \snrc=r, \quad -\lim_{\snrc\to\infty}\log P_\zz{e}/\log \snrc=d.
\end{equation}
In \ac{mimo} systems, the fundamental interplay between these two gains is characterized by the \ac{dmt}. Let $d^*(r)$ denote the optimal \ac{dmt} curve, which represents the supremum of achievable diversity gains for any target multiplexing gain $r$.

\begin{figure}[t]
    \centering
    \includegraphics[width=0.7\columnwidth]{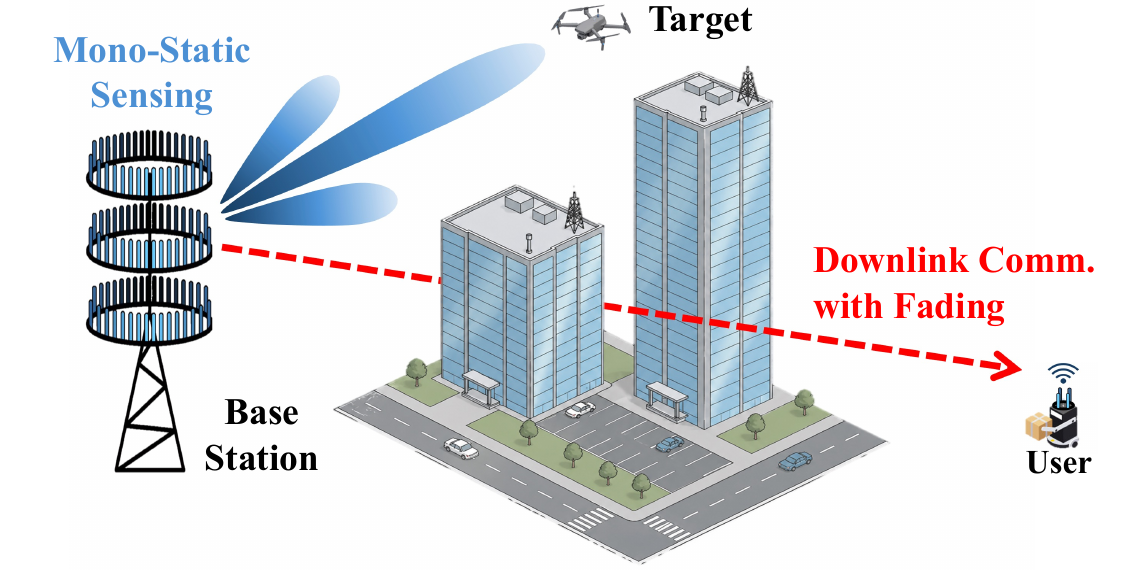}
    \vspace{-0.5em}
    \caption{Illustration of a typical mono-static \ac{mimo} \ac{isac} scenario.}
    \label{fig:system}
    \vspace{-1em}
\end{figure}

\subsection{\ac{sac} Optimal Waveform Design}
Next, we consider the dedicated optimal waveforms that strictly focus on either the sensing channel \eqref{eq:channelmodel_s} or the communication channel \eqref{eq:channelmodel_c}. For sensing, the optimal waveform aims to minimize the \ac{bcrb} $e$. As demonstrated in \cite{XioLiuCui:J23}, this metric is determined solely by the sample covariance matrix of the transmitted signal:
\begin{equation}
    \RM{R}_{\RM{X}}\triangleq T^{-1}\RM{X}\RM{X}\ermtrans.
\end{equation}
Under mild regularity conditions \cite[Prop. 3-5]{XioLiuCui:J23}, the optimal sensing covariance is deterministic, i.e., $\RM{R}_{\RM{X}}=\M{R}$. To facilitate our analysis, we always assume the existence of a unique optimal $\M{R}$. Consequently, any sensing-optimal waveform must reside within the set $\{\M{X}:\M{X}\M{X}\ermtrans=T\M{R}\}$.

Conversely, for the communication link, the optimal waveform achieves the unconstrained \ac{dmt} $d^*(r)$. Provided that $T\geq M+\nc-1$, this optimal tradeoff is achieved by isotropic Gaussian codewords, i.e., $\text{vec}(\RM{X})\sim \mathcal{CN}(\V{0},\M{I}_{MT})$. The resulting optimal \ac{dmt} curve $d^*(r)$ is the well-known piecewise-linear function \cite{ZheTse:J03} connecting the points $(k,d^*(k))$ for $k=0,...,\min\{M,\nc\}$, where
\begin{equation}\label{eq:originaldmt}
    d^*(k)=(M-k)(\nc-k).
\end{equation}
By substituting this communication-optimal Gaussian waveform into \eqref{eq:bfim} and \eqref{eq:MillerChang}, its corresponding sensing performance can be readily evaluated. 
However, characterizing the communication \ac{dmt} when strictly employing sensing-optimal codewords is highly non-trivial. In the remainder of this paper, we leverage Riemannian geometry to tackle this challenging problem, ultimately deriving a converse bound on the sensing-constrained \ac{dmt}, denoted as $d^*_{\M{R}}(r)$.

\section{\ac{dmt} with Optimal Sensing Waveform}
In this section, we derive a converse bound on the sensing-constrained \ac{dmt} $d^*_{\M{R}}(r)$ via outage analysis under $\M{X}\M{X}\ermtrans=T\M{R}$. This constraint defines a smooth manifold in $\Complex^{M\times T}$, a generalized Stiefel manifold defined below:
\begin{definition}\label{def:stiefel}
    For a Hermitian matrix $\M{0}\preceq\M{A}\in \Complex^{k\times k}$ and an integer $n
    \geq\zz{rk}(\M{A})$, a generalized Stiefel manifold is
    \begin{equation}
        S_{\M{A}}^{(k,n)}\triangleq\left\{\M{X}\in\Complex^{k\times n};\M{X}\M{X}\ermtrans=\M{A}\right\}
    \end{equation}
    which forms a smooth real manifold of dimension $\zz{rk}(\M{A})(2n-\zz{rk}(\M{A}))$. Notably, the standard complex Stiefel manifold \cite{Chi:B03} is recovered as the special case $S_{\M{I}_k}^{(k,n)}$.
\end{definition}

For subsequent analysis, we also recall the unitary Lie group of degree $n$, given by 
$
    U(n)\triangleq\{\M{U}\in\Complex^{n\times n} \mid \M{U}\M{U}\ermtrans=\M{I}_{n}\}=S_{\M{I}_n}^{(n,n)}.
$
Crucially, $U(n)$ exerts a smooth and transitive group action on $S_{\M{A}}^{(k,n)}$ via right matrix multiplication.

In this work, we endow the generalized Stiefel manifolds with the standard Euclidean metric. That is, for any tangent vectors $\M{\Delta}_1,\M{\Delta}_2$ in the tangent space $\zz{T}_{\M{S}}S_{\M{A}}^{(k,n)}$ of $\M{S}\in S_{\M{A}}^{(k,n)}$, the metric tensor is given by $g(\M{\Delta}_1,\M{\Delta}_2)\triangleq\Re\{\text{tr}(\M{\Delta}_1 \M{\Delta}_2\ermtrans)\}$. The corresponding Euclidean volume measure on $S_{\M{A}}^{(k,n)}$ is denoted as $\zz{Vol}(\cdot)$.

Generalized Stiefel manifolds possess several pivotal geometric properties that underpin our analysis:
\begin{proposition}[Properties of $S_{\M{A}}^{(k,n)}$]\label{prop:stielfelproperty}
    \begin{enumerate}[(a)]
        \item $S_{\M{A}}^{(k,n)}$ is a $U(n)$-homogeneous space. The metric $g(\cdot,\cdot)$ and measure $\zz{Vol}(\cdot)$ are $U(n)$-invariant. The uniform probability measure $P_\zz{H}(\cdot)\triangleq \zz{Vol}(\cdot)/\zz{Vol}(S_{\M{A}}^{(k,n)})$ is the unique $U(n)$-invariant probability measure (Haar measure).\label{prop:homogeneous}
        \item Let random matrix $\RM{F}$ distributed according to $P_\zz{H}$ on $S_{\M{A}}^{(k,n)}$. For any $\M{K}\in \Complex^{m\times k}$, $\M{K}\RM{F}\sim P_\zz{H}$ on $S_{\M{K}\M{A}\M{K}\ermtrans}^{(m,n)}$.\label{prop:invariancestay}
    \end{enumerate}
\end{proposition}
\begin{proof}
See Appendix A.
\end{proof}

With the geometric perspective, to transmit sensing-optimal codewords is to transmit on a generalized Stiefel manifold and, given a channel realization $\M{H}_\zz{c}$, the noiseless received communication signal also lies on such a manifold, i.e.
\begin{equation}\label{eq:codewordconstraint}
    \RM{X}\in S_{T\M{R}}^{(M,T)},\quad \sqrt{{\snrc}/{M}}\M{H}_\zz{c}\RM{X}\in S_{{\snrc T}{\M{H}_\zz{c}R\M{H}_\zz{c}\ermtrans}/M}^{(N_\zz{c},T)}.
\end{equation}
To derive a converse bound on the \ac{dmt} under optimal sensing constraints, we follow the methodology established in the seminal work \cite{ZheTse:J03} by first characterizing the outage probability in the high-\ac{snr} regime.

\subsection{Outage Characterization}
In this subsection, we analyze the communication outage for the quasi-static fading channel
\begin{equation}
    \label{eq:compoundch}
    \RM{Y}_{\zz{c},t}=\sqrt{\snrc/M}\RM{H}_\zz{c}\RM{X}_t+\RM{Z}_\zz{c}, \quad t=1,2,\dots
\end{equation}
where infinite blocks are transmitted via the same fading channel $\RM{H}_\zz{c}$, and each block is constrained on the sensing-optimal manifold, i.e., $\RM{X}_t\in S_{T\M{R}}^{(M,T)}$. Omitting subscript $t$ for brevity, the outage probability at target rate $R$ is:
\begin{equation}\label{eq:outageformulation}
    P_\zz{out}(R)=\min_{P_{\RM{X}};\RM{X}\indep\RM{H}_\zz{c}}\left\{\mathbb{P}_{\RM{H}_\zz{c}}\left(I_{\RM{H}_\zz{c}}(\RM{X};\RM{Y}_\zz{c})<TR\right)\right\}.
\end{equation}
where $\indep$ denotes the independence of random variables. $I_{\M{H}_\zz{c}}(\RM{X};\RM{Y}_\zz{c})$ signifies the mutual information conditioned on a fixed realization $\M{H}_\zz{c}$, whereas $I_{\RM{H}_\zz{c}}(\RM{X};\RM{Y}_\zz{c})$ represents the associated random variable\footnote{This notational convention for distinguishing between specific realizations and their corresponding random variables is maintained throughout the remainder of this paper without further explicit mention.}. 

To calculate \eqref{eq:outageformulation}, we have
\begin{equation}\label{eq:mutualinformation}
\begin{split}
    &I_{\M{H}_\zz{c}}(\RM{X};\RM{Y}_\zz{c})
    =h_{\M{H}_\zz{c}}(\RM{Y}_\zz{c})-h_{\M{H}_\zz{c}}(\RM{Y}_\zz{c}|\RM{X})\\
    &=h\left(\sqrt{{\snrc}/{M}}\M{H}_\zz{c}\RM{X}+\RM{Z}_\zz{c}\right)-N_\zz{c}T\log (\pi e).
\end{split}
\end{equation}
Note that $\sqrt{{\snrc}/{M}}\M{H}_\zz{c}\RM{X}\in S_{{\snrc T}\M{\M{H}_\zz{c}R\M{H}_\zz{c}\ermtrans}/M}^{(N_\zz{c},T)}$ by \eqref{eq:codewordconstraint}. 

Using Proposition \ref{prop:stielfelproperty}, since $S_{{\snrc T}\M{\M{H}_\zz{c}R\M{H}_\zz{c}\ermtrans}/M}^{(N_\zz{c},T)}$ is homogeneous, uniform distribution $P_\zz{H}$ should maximize $I_{\M{H}_\zz{c}}(\RM{X};\RM{Y}_\zz{c})$, which can be achieved by choosing $\RM{X}\sim P_\zz{H}$ regardless of $\M{H}_\zz{c}$. Following this intuition, we arrive at the following proposition.
\begin{proposition}\label{prop:maxmutual}
    The maximum mutual information \eqref{eq:mutualinformation} and the outage probability \eqref{eq:outageformulation} are obtained by $P_\zz{H}$, i.e., 
    $P_\zz{H}=\underset{P_{\RM{X}}}{\zz{argmax}}\left[I_{\M{H}_\zz{c}}(\RM{X};\RM{Y}_\zz{c})\right]=\underset{P_{\RM{X}};\RM{X}\indep\RM{H}_\zz{c}}{\zz{argmin}}\left\{\mathbb{P}_{\RM{H}_\zz{c}}\left(I_{\RM{H}_\zz{c}}(\RM{X};\RM{Y}_\zz{c})<TR\right)\right\}$.\footnote{This result is also latently used in \cite{XioLiuCui:J23} as a conjecture. }   
\end{proposition}
\begin{proof}
See Appendix B.
\end{proof}

For a smooth Euclidean submanifold, \cite{Wei:T17} finds that at high SNR, the optimal distribution is asymptotically uniform. In our case, we proved $P_\zz{H}$ achieves the maximal entropy regardless of SNR level. Therefore, we always set $\RM{X}\sim P_\zz{H}$ on $S_{T\M{R}}^{(M,T)}$. 

Next, we apply the \ac{svd} to decompose $\M{H}_\zz{c}\RM{X}$, effectively transforming the Rx manifold into a canonical form. For brevity, we consider the case where $\rank{R} \leq \nc$ and $T \geq \nc$.

By compact \ac{svd}, we have the decomposition $\M{H}_\zz{c}\RM{X}/\sqrt{T}=\M{V}\M{\Sigma}\RM{W}$, with $\M{V}\in U(\nc)$, $\M{\Sigma}\in \Real^{\nc\times\nc}$ and $\RM{W}\in S_{\M{I}_{\nc}}^{(\nc,T)}$. Let $\M{\Sigma}=\zz{diag}(\sigma_1,...,\sigma_{\rank{R}},\V{0})\triangleq \zz{diag}(\V{\sigma}\trans,\V{0})$ and $\sigma_1\geq...\geq\sigma_{\rank{R}}\geq0$. $\M{V}, \M{\Sigma}$ are functions of $\hcnb$ by eigendecomposition $\hcnb\M{R}\hcnb\ermtrans=\M{V}\M{\Sigma}^2\M{V}\ermtrans$. Since $|\det\M{V}|=1$, we have
\begin{subequations}
\begin{equation}
        h\left(\sqrt{{\snrc}/{M}}\M{H}_\zz{c}\RM{X}+\RM{Z}_\zz{c}\right)
        =h\left({\sqrt{{\snrc T }/{M}}\M{\Sigma}\RM{W}}+\RM{Z}_\zz{c}\right)\label{eq:rxentropy}
\end{equation}
\begin{equation}\label{eq:XequalsC}
    I_{\hcnb}({\RM{X}};\RM{Y}_\zz{c})=I_{\V{\sigma}}(\RM{C}+\RM{Z}_\zz{c};\RM{W})
\end{equation}    
\end{subequations}
where $\RM{C} \triangleq \sqrt{\snrc T / M}\M{\Sigma}\RM{W}$. By Proposition \ref{prop:stielfelproperty}-b, the uniform distribution of $\RM{X}$ ensures that $\RM{W}$ and $\RM{C}$ are uniformly distributed on $S_{\M{I}_{\nc}}^{(\nc,T)}$ and $S_{\snrc T\M{\Sigma}^2/M}^{(\nc,T)}$, respectively. Consequently, the mutual information and the outage probability \eqref{eq:outageformulation} are entirely characterized by the geometry of the manifold $S_{\snrc T\M{\Sigma}^2/M}^{(\nc,T)}$, whose ``shape'' depends on the SNR $\snrc$ and the singular values $\V{\sigma}$.

To facilitate geometric intuition, we define the log-singular value $\alpha_i\triangleq -\log \sigma_i/\log \snrc$, and use $\V{\alpha}\triangleq(\alpha_1,...,\alpha_{\rank{R}})$, $ \alpha_1\leq...\leq\alpha_{\rank{R}}$. By the definition of $\RM{C}$, we have
\begin{equation}\label{eq:CW}
    \RM{C}=\sqrt{{T}/{M}}
    \zz{diag}\left(\snrc^{0.5-\alpha_1},...,\snrc^{0.5-\alpha_{\rank{R}}}, 0,...,0\right)
    \RM{W}.
\end{equation}
Since $\RM{W}$ is uniform on the standard Stiefel manifold $S_{\M{I}_{\nc}}^{(\nc,T)}$, \eqref{eq:CW} characterizes an ``inflation'' process in $\Complex^{\nc\times T}$. Specifically, while the rising \ac{snr} drives a polynomial expansion of the generalized Stiefel manifold with an exponential order of $0.5$, the channel $\M{H}_\zz{c}$ moderates this growth by reducing the expansion order by $\alpha_i$. When the composite channel $\M{H}_\zz{c}\RM{X}$ exhibits atypically small singular values (i.e., high $\alpha_i$), the manifold $S_{\snrc T\M{\Sigma}^2/M}^{(\nc,T)}$ wouldn't be inflated enough to support reliable communication, that is when the outage event appears, i.e., the channel realization $\hcnb$ is in ``deep fade'' \cite{ZheTse:J03}.

By \eqref{eq:XequalsC}, the outage probability \eqref{eq:outageformulation} can be reformulated as
\begin{equation}\label{eq:outagereformulation}
    P_\zz{out}(R)=\mathbb{P}_{\RV{\alpha}}\left(I_{\RV{\alpha}}(\RM{C}+\RM{Z}_\zz{c};\RM{W})<TR\right).
\end{equation}
To derive an asymptotic result for \eqref{eq:outagereformulation}, we first characterize $I_{\V{\alpha}}(\RM{C}+\RM{Z}_\zz{c};\RM{W})$ with a fixed $\V{\alpha}$ in the high-\ac{snr} regime.

\subsection{Asymptotic Mutual Information}
In this subsection, we characterize the asymptotic mutual information for a fixed realization $\V{\alpha}$ as $\snrc \to \infty$. Per \eqref{eq:CW}, the manifold expansion rate is governed by the exponents $\{0.5-\alpha_i\}$. For dimensions where $\alpha_i \geq 0.5$, the manifold fails to expand w.r.t. $\snrc$, implying that effective communication is predominantly supported by the subset of dimensions satisfying $\alpha_i < 0.5$. 
We define
$
    m \triangleq \left|\{i \mid \alpha_i < 0.5\}\right|
$
as the number of such dominant dimensions. Since $m=0$ results in a trivially bounded mutual information, we focus on $m > 0$ throughout this section, where $\alpha_m$ denotes the largest log-singular value satisfying $\alpha_i < 0.5$.
Our geometric intuition is formalized in the following proposition.
\begin{proposition}\label{prop:onehalf}
    Given $\V{\alpha}$, define $\RM{W}_1\sim P_\zz{H}$ on $S_{\M{I}_m}^{(m,T)}$ and let $\RM{Z}_\zz{c1}\in\Complex^{m\times T}$ be a noise matrix with i.i.d. $\mathcal{CN}(0,1)$ entries. The following asymptotic equivalence holds:
    \begin{equation}
        \lim_{\snrc\to\infty}\frac{I_{\V{\alpha}}\left(\sqrt{{T}/{M}}
        \zz{diag}\left(\snrc^{0.5-\V{\alpha}_{1:m}}\right)\RM{W}_1+\RM{Z}_\zz{c1};\RM{W}_1\right)}{I_{\V{\alpha}}(\RM{C}+\RM{Z}_\zz{c};\RM{W})}=1.
    \end{equation}
\end{proposition}
\begin{proof}
See Appendix C.
\end{proof}

\begin{figure*}[th]
\begin{equation}\label{eq:cdefinition}
    c_{\snrc,\V{\alpha}}\triangleq\max\left\{
    \max_{\M{S}\in S_{\snrc,\V{\alpha}}, \Delta\in \zz{T}_{\M{S}}S_{\snrc,\V{\alpha}},g(\Delta,\Delta)=1} \left|\mathbf{II}_{\M{S}}(\Delta,\Delta)\right|,
    \rho_{\perp}\left(S_{\snrc,\V{\alpha}}\right)^{-1},
    \rho_{\top}\left(S_{\snrc,\V{\alpha}}\right)^{-1}
    \right\}
\end{equation}
\begin{equation}\label{eq:boundedentropy}
    \left|h\left(\RM{C}_1+\RM{Z}_\zz{c1}\right)-{m^2}\log(\pi e)/2-\log \zz{Vol}\left(S_{\snrc,\V{\alpha}}\right)\right|\leq \zz{const}(m,T)\delta^{-1}(1+c_{\snrc,\V{\alpha}}\snrc^{0.5})^{\delta}c_{\snrc,\V{\alpha}}^2\log^2(c_{\snrc,\V{\alpha}})
\end{equation}
\vspace{-1em}
\rule{\linewidth}{0.5pt}
\end{figure*}

Proposition \ref{prop:onehalf} implies that the asymptotic mutual information is determined solely by the log-singular values $\V{\alpha}_{1:m}$, corresponding to the first $m$ rows of $\RM{C}$ that expand at a positive rate according to \eqref{eq:CW}. Define the reduced-dimension noiseless component as
\begin{equation}
     \RM{C}_1\triangleq \sqrt{T/M} \zz{diag}(\snrc^{0.5-\V{\alpha}_{1:m}}) \RM{W}_1
\end{equation}
which is uniformly distributed on the manifold $S_{\snrc,\V{\alpha}}\triangleq S_{{\frac{T}{M}}{\zz{diag}\left(\snrc^{1-2\V{\alpha}_{1:m}}\right)}}^{(m,T)}$ by Proposition \ref{prop:stielfelproperty}-b. In the high-\ac{snr} regime with $R=r\log\snrc$, the asymptotic outage probability \eqref{eq:outagereformulation} satisfies
\begin{equation}\label{eq:outageasymp}
    P_\zz{out}(R) 
    \doteq \mathbb{P}_{\RV{\alpha}}\left(h_{\RV{\alpha}}\left(\RM{C}_1+\RM{Z}_\zz{c1}\right) < TR\right).
\end{equation}
As shown in \eqref{eq:outageasymp}, the asymptotic outage is governed by the differential entropy of a uniform distribution over the inflating manifold $S_{\snrc,\V{\alpha}}$ perturbed by normalized additive noise $\RM{Z}_\zz{c1}$.

Previous literature, such as \cite{XioLiuCui:J23}, utilized Weyl's tube formula \cite{Wey:J39} to approximate similar but simpler entropy terms. Specifically, \cite[Thm. 1]{XioLiuCui:J23} characterized the maximal differential entropy $h_{\hcnb}(\RM{C}_\zz{x}+\RM{Z}_\zz{cx})$ for $\RM{C}_\zz{x}\in S_{\snrc T\hcnb\M{R}\hcnb\ermtrans/M}^{(\nc,T)}$ as $\snrc\to\infty$. This was achieved by approximating the noise distribution normal to the manifold and assuming uniformity on the surface of every $\epsilon$-tube. However, since $S_{\snrc,\V{\alpha}}$ exhibits row-dependent expansion rates $\snrc^{0.5-\alpha_i}$, such approximations are inadequate for our rigorous analysis. We thus leverage the formal results in \cite{Wei:T17}, which characterize the entropy of manifold-constrained variables via tubular neighborhood theory. This leads to the following lemma.

\begin{lemma}\label{lemma:differentialentropy}
Define $c_{\snrc,\V{\alpha}}$ as in \eqref{eq:cdefinition}, where $\mathbf{II}_{\M{S}}(\cdot,\cdot)$ is the second fundamental form under Euclidean embedding, $\rho_{\perp}(\cdot)$ denotes the manifold's maximal uniform tubular neighborhood radius whereas $\rho_{\top}(\cdot)$ is the injectivity radius. For any $\delta\in(0,1]$, inequality \eqref{eq:boundedentropy} holds.
\end{lemma}
\begin{proof}
    This lemma follows from \cite[Thm. 4.2.1]{Wei:T17} by noting that the p.d.f. of $\RM{C}_1$ relative to the volume measure $\zz{Vol}(\cdot)$ is the constant $\left(\zz{Vol}(S_{\snrc,\V{\alpha}})\right)^{-1}$.
\end{proof}

Lemma \ref{lemma:differentialentropy} establishes that the target entropy is asymptotically approximated by the manifold-entropy $\log \zz{Vol}(S_{\snrc,\V{\alpha}})$, plus a term $\frac{1}{2}m^2\log(\pi e)$ representing the noise contribution within the manifold's normal bundle. The approximation error is bounded by the RHS of \eqref{eq:boundedentropy}, which is determined by $\snrc$ and the geometric parameter $c_{\snrc,\V{\alpha}}$ in \eqref{eq:cdefinition}. $c_{\snrc,\V{\alpha}}$ integrates three critical geometric properties: the maximum principal curvature $\max|\mathbf{II}_{\M{S}}(\Delta,\Delta)|$, the injectivity radius $\rho_{\top}(S_{\snrc,\V{\alpha}})$, and the maximal uniform tubular neighborhood radius $\rho_{\perp}(S_{\snrc,\V{\alpha}})$. In the following theorem, we show that these quantities vanish at a rate no slower than $\snrc^{\alpha_m-0.5}$. By selecting an appropriate $\delta$, the error term in \eqref{eq:boundedentropy} becomes negligible as $\snrc \to \infty$.
\begin{theorem}[Asymptotic Mutual Information]\label{theorem:asyentropy}
Given $\V{\alpha}$, the geometric parameter $c_{\snrc,\V{\alpha}}$ satisfies
\begin{equation}
    c_{\snrc,\V{\alpha}}\overset{.}{\leq}\snrc^{\alpha_m-0.5}.
\end{equation}
The corresponding asymptotic differential entropy is given by
\begin{equation}
    \lim_{\snrc\to\infty}\left|h_{\V{\alpha}}\left(\RM{C}_1+\RM{Z}_\zz{c1}\right)-\frac{m^2}{2}\log(\pi e)-\log \zz{Vol}\left(S_{\snrc,\V{\alpha}}\right)\right|=0.
\end{equation}
We obtain the following scaling result:
\begin{align}
    I_{\hcnb}({\RM{X}};\RM{Y}_\zz{c}) &\sim h_{\V{\alpha}}(\RM{C}_1+\RM{Z}_\zz{c1}) \nonumber \\
    &\sim \sum_{i=1}^{\rank{R}}(2T+1-2i)(0.5-\alpha_i)^+\log(\snrc).
\end{align}
\end{theorem}
\begin{proof}
See Appendix D.
\end{proof}

\begin{remark}
Theorem \ref{theorem:asyentropy} establishes the asymptotic mutual information of the compound channel \eqref{eq:compoundch} via geometric analysis. For a given realization $\V{\alpha}$, the mutual information scales with $\log\snrc$ with a coefficient of $\sum_{i=1}^{\rank{R}}(2T+1-2i)(0.5-\alpha_i)^+$. Consistent with our intuition, only the log-singular values satisfying $\alpha_i < 0.5$ contribute to the asymptotic growth of $I_{\hcnb}(\RM{X};\RM{Y}_\zz{c})$. In the following section, we leverage the distribution of $\RV{\alpha}$ to evaluate the asymptotic outage probability and derive a converse bound on the sensing-constrained \ac{dmt}.
\end{remark}

\subsection{Converse Bound on the Sensing-Constrained \ac{dmt}}
In this section, we first determine the asymptotic outage probability with rate $r\log \snrc$ as $\snrc\to\infty$, i.e.,
\begin{equation}\label{eq:outagerereformulation}
    P_\zz{out}(r\log \snrc)\doteq\mathbb{P}_{\RV{\alpha}}\left(I_{\RV{\alpha}}(\RM{C}_\zz{1}+\RM{Z}_\zz{c1};\RM{W}_\zz{1})<rT\log\snrc\right)
\end{equation}
Applying Theorem \ref{theorem:asyentropy} and neglecting boundary sets of measure zero, the outage condition in \eqref{eq:outagerereformulation} defines the following region $\mathcal{A}\subset\Real^{\rank{R}}$ in the $\V{\alpha}$-domain:
\begin{equation}
   \mathcal{A}\triangleq\{
    \alpha_i<\alpha_{i+1},
    \sum\nolimits_{i=1}^{\rank{R}}(2T+1-2i)(0.5-\alpha_i)^+<Tr\}
\end{equation}
and the asymptotic outage probability is thus:
\begin{equation}\label{eq:outageevent}
    \vspace{-0.2em}
    P_\zz{out}(r\log \snrc)\doteq\int_\mathcal{A}p_{\RV{\alpha}}(\V{\alpha})\zz{d}\V{\alpha}.
    \vspace{-0.2em}
\end{equation}
The following lemma provides the joint p.d.f. $p_{\RV{\alpha}}(\V{\alpha})$ required to evaluate this integral.

\begin{lemma}\label{lemma:singularvalue}
    Let $\lambda_1>\lambda_2>...>\lambda_{\rank{R}}>0$
    \footnote{The eigenvalue distribution of $\RM{H}_\zz{c}\M{R}\RM{H}_\zz{c}\ermtrans$ is continuous w.r.t. $\{\lambda_i\}$ \cite{Tul:B04}. Consequently, \eqref{eq:alphadistributionasymp} readily extends to the case of repeated eigenvalues.}
    be the eigenvalues of the transmit covariance matrix $\M{R}$, the joint distribution of $\RV{\alpha}=(\rv{\alpha}_1,\rv{\alpha}_2,...,\rv{\alpha}_{\rank{R}})\trans$ is:
    \begin{align}\label{eq:alphadistribution}
        &p_{\RV{\alpha}}(\V{\alpha})=F(\snrc, \rank{R},\nc,\M{\Lambda}_{>0})\det\left(\left\{{e^{-\snrc^{-2\alpha_i}/{\lambda_j}}}\right\}\right)\nonumber\\
        &\prod_{s=1}^{\rank{R}}\snrc^{-2\alpha_s(\nc-\rank{R}+1)}\prod_{k<l}^{\rank{R}}\left(\snrc^{-2\alpha_k}-\snrc^{-2\alpha_l}\right)
    \end{align}
    where $F(\snrc, \rank{R},\nc,\M{\Lambda}_{>0})$ is the normalizing constant. As $\snrc\to\infty$, $p_{\RV{\alpha}}(\V{\alpha})$ has the following asymptotic expression
    \begin{equation}\label{eq:alphadistributionasymp}
    \begin{split}
        &\lim_{\snrc\to\infty}\left[\log{p_{\RV{\alpha}}(\V{\alpha})}/{\log\snrc}\right]=\\
        &\begin{cases}
        -2\sum_{i=1}^{\rank{R}}(\nc+\rank{R}+1-2i)\alpha_i, & \alpha_{1}>0\\
        -\infty, & \alpha_{1}<0
        \end{cases}
    \end{split}
    \end{equation}
\end{lemma}
\begin{proof}
    See Appendix E.
\end{proof}

Lemma \ref{lemma:singularvalue} indicates that if any $\alpha_i$ falls below zero, the p.d.f. decays super-polynomially. Thus, neglecting the zero-measure boundaries, we focus on the positive orthant $\mathcal{A}' \triangleq \mathcal{A} \cap (\Real^+)^{\rank{R}}$. The asymptotic outage probability \eqref{eq:outageevent} satisfies
\begin{equation}\label{eq:outageanddout}
\begin{split}
        P_\zz{out}&(r\log \snrc)\doteq\int_{\mathcal{A}'}\snrc^{-2\sum_{i=1}^{\rank{R}}(\nc+\rank{R}+1-2i)\alpha_i}\zz{d}\V{\alpha}.\\
        &\doteq\snrc^{-2\underset{\V{\alpha}\in\mathcal{A}'}{\inf}\left(\sum_{i=1}^{\rank{R}}(\nc+\rank{R}+1-2i)\alpha_i\right)}\triangleq \snrc^{-d^\zz{out}_{\rank{R}}}
\end{split}
\end{equation}
where the second equality utilized Laplace's principle \cite{Dem:B09}.

Finally, solving for $d^{\zz{out}}_{\rank{R}}$ in \eqref{eq:outageanddout} characterizes the tradeoff between the asymptotic outage probability $P_\zz{out} \doteq \snrc^{-d^{\zz{out}}_{\rank{R}}}$ and the rate $r\log\snrc$ for the compound channel \eqref{eq:compoundch}. This simultaneously provides a converse bound on the sensing-constrained \ac{dmt}, leading to the main result of this paper.

\begin{theorem}\label{theo:final}
    The outage exponent $d^{\zz{out}}_{\rank{R}}(r)$ is the piecewise-linear function connecting the points $(r(k), d^{\zz{out}}_{\rank{R}}(r(k)))$ for $k \in \{0, \dots, \rank{R}\}$, where
    \begin{equation}\label{eq:newdmt}
        (r(k),d^{\zz{out}}_{\rank{R}}(r(k)))=\big( k\big( 1-\frac{k}{2T}\big),(\nc-k)(\rank{R}-k)\big).
    \end{equation}
    The function $d^{\zz{out}}_{\rank{R}}(r)$ serves as a converse bound for the sensing-constrained \ac{dmt} $d^*_{\M{R}}(r)$, i.e.,
    \begin{equation}\label{eq:newdmtupbound}
        d^*_{\M{R}}(r) \leq d^{\zz{out}}_{\rank{R}}(r).
    \end{equation}
\end{theorem}
\begin{proof}
    See Appendix F.
\end{proof}

\begin{figure}[t]
	\centering
    \includegraphics[width=0.7\columnwidth]{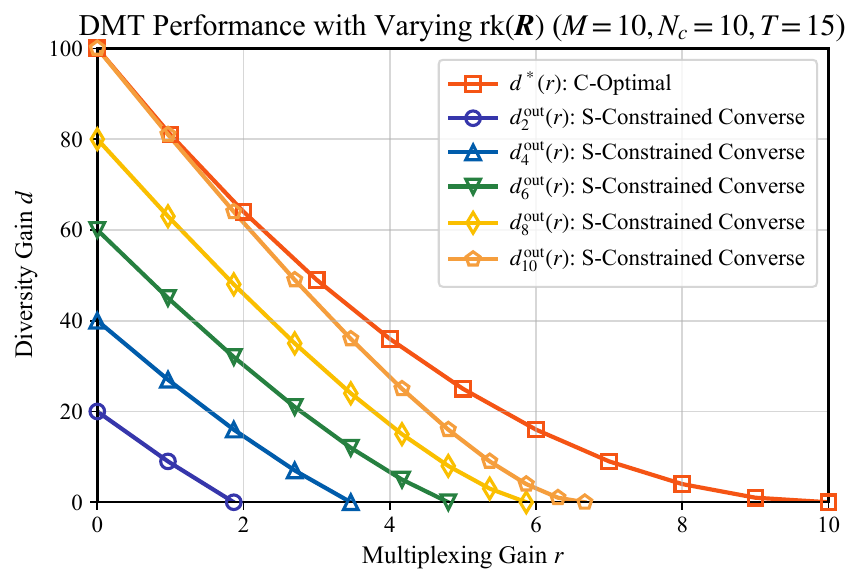}
    \vspace{-1em}
    \caption{Impact of $\rank{\M{R}}$ on the sensing-constrained \ac{dmt}. The converse bound for $d^*_{\M{R}}(r)$ also serves as the converse bound in the presence of $N_\zz{t} = \rank{\M{R}}$ sensing targets in Section IV-A.}
    \vspace{-1em}
    \label{fig:DMTVaryRk}
\end{figure}

\addtolength{\topmargin}{.05in}

To visualize the implications of Theorem \ref{theo:final}, we plot the derived converse bound alongside the original unconstrained communication \ac{dmt} \eqref{eq:originaldmt} in Fig. \ref{fig:DMTVaryRk}. As illustrated, our sensing-constrained bound lies strictly within the original \ac{dmt} curve. While structurally similar to the piecewise-linear formulation of \eqref{eq:originaldmt}, Theorem \ref{theo:final} reveals a fundamental shift: the performance is dictated by the effective rank $\rank{R}$ rather than the physical transmit antenna count $M$. Consequently, for a fixed $\rank{R}$, scaling up $M$ yields no additional diversity or multiplexing gains. In the full-rank regime ($\rank{R}=M$), where the bound reaches its maximum, the diversity gain $(\nc-k)(M-k)$ for a given $k$ matches the unconstrained case \eqref{eq:originaldmt}. However, the multiplexing gain incurs a penalty of $k^2/2T$, fundamentally reflecting the intrinsic loss in communication \ac{dof} necessitated by the sensing constraint.

Evaluating these extreme points yields further insights. At the diversity-optimal point $(0,\nc\rank{R})$, the system incurs no penalty compared to an unconstrained $\rank{R}$-antenna \ac{mimo} setup. Conversely, the multiplexing-optimal point is $\left(\rank{R}\left(1-\frac{\rank{R}}{2T}\right),0\right)$, which aligns perfectly with conventional \ac{dof} analysis. Because the transmit signal $\RM{X}$ resides on the generalized Stiefel manifold $S_{T\M{R}}^{(M,T)}$ with $\rank{R}(2T-\rank{R})$ real dimensions (\ac{dof}), and each real \ac{dof} contributes $1/2T$ to the capacity pre-log factor, our derived multiplexing gain is exactly recovered. This geometric consistency corroborates the tightness of our converse bound.

\begin{remark}
Ultimately, Theorem \ref{theo:final} gives a preliminary answer to the pivotal question posed at the outset. The loss in \ac{mimo} gain from transmitting sensing-optimal codewords is twofold: First, the sensing constraint fundamentally restricts the system to $\rank{R}$ effective Tx dimensions, limiting the entire \ac{dmt}. Second, compared to an unconstrained system with $M=\rank{R}$, although the full spatial diversity is preserved, the multiplexing gain incurs a penalty that grows at higher target rates.
Due to space limits, we assume $\rank{R}\leq \nc$, $T\geq \nc$. However, our framework readily generalizes by replacing $\rank{R}$ with $\min\{\nc,\rank{R}\}$. We defer this extension and the achievability bound to future work.
\end{remark}

\section{Case Study}
This section conducts two case studies to gain deeper insights into sensing-constrained \ac{mimo} \ac{isac} systems.

\subsection{Angle Estimation of Multiple Targets}
In this scenario, a sensing \ac{bs} equipped with $M=10$ transmit antennas performs angular estimation on $N_\zz{t}$ targets while simultaneously communicating with a \ac{ue} having $\nc=10$ receive antennas. The sensing channel is modeled as
\vspace{-0.5em}
\begin{equation}\label{eq:casestudya}
\vspace{-0.5em}
\RM{H}_\zz{s}=\sum\nolimits_{n=1}^{N_\zz{t}}\rv{\beta}_n \V{a}(\rv{\theta}_n)\V{v}\ermtrans(\rv{\theta}_n)
\end{equation}
where $\rv{\beta}_n$ denotes the radar cross-section of the $n$-th target, $\{\rv{\theta}_n\}_{n=1}^{N_\zz{t}}$ represents the targets' angular positions, and $\V{v}(\rv{\theta}_n)$ and $\V{a}(\rv{\theta}_n)$ are the steering vectors of the transmit and sensing receive arrays, respectively. The parameter vector to be estimated is $\RV{\eta}=(\rv{\theta}_1, \dots, \rv{\theta}_{N_\zz{t}})$, with dimension $K=N_\zz{t}$. According to \cite[Co. 2]{XioLiuCui:J23}, the $\rank{R}$ satisfies $\rank{R}\leq \min\{M,N_\zz{t}\}$. Therefore, the sensing-constrained \ac{dmt} is upper-bounded by $d^*_{\M{R}}(r)\leq d^{\zz{out}}_{\rank{R}}(r) \leq d^{\zz{out}}_{N_\zz{t}}(r)$. 
Fig. \ref{fig:DMTVaryRk} plots the converse bound $d^{\zz{out}}_{N_\zz{t}}(r)$ across $N_\zz{t} \in \{2,4,6,8,10\}$. Counterintuitively, our \ac{dmt} bound improves as $N_\zz{t}$ increases. This phenomenon arises because Rayleigh fading channels inherently benefit from spatially diverse transmissions. While tracking fewer targets restricts the transmit beam to a limited set of spatial directions, a larger target count forces a more uniform spatial spectrum, thereby unlocking additional communication \ac{dof}.

\subsection{$\RM{H}_\zz{s}$ Estimation}
In this scenario, we consider a system with $M=\nc=3$ and a varying blocklength $T$, where the sensing \ac{bs} aims to estimate the sensing channel matrix $\RM{H}_\zz{s}$ directly. Assuming $\mathrm{vec}(\RM{H}_\zz{s}) \sim \mathcal{CN}(\V{0}, \M{I}_{MN_\zz{s}})$, it follows from \cite{XioLiuCui:J23} that the optimal covariance matrix is $\M{R}=\M{I}$ with $\rank{R}=3$. The corresponding sensing-constrained \ac{dmt} bound is illustrated in Fig. \ref{fig:DMTVaryT}. As $T$ increases, the constrained \ac{dmt} gradually approaches the original \ac{dmt}. This behavior can also be understood from a \ac{dof} perspective: transmitting over the space $S_{\M{R}}^{(M,T)}$ incurs a penalty of $M^2$ \acp{dof} per block. Consequently, the \ac{dof} loss per symbol is $M^2/T$. As $T$ grows, this fractional loss diminishes, which increases the available communication \ac{dof} per symbol and ultimately yields a higher \ac{dmt}.

This observation highlights another fundamental tradeoff unique to \ac{mimo} \ac{isac} systems: the tension between latency and \ac{mimo} gain. 
Unlike the original \ac{dmt} \eqref{eq:originaldmt}, which is independent of the blocklength $T$, our bound reveals that when optimal-sensing capabilities are incorporated, achieving lower latency sacrifices \ac{mimo} gain. Conversely, in the latency-unconstrained regime as $T \to \infty$, \eqref{eq:newdmt} shows that the sensing-constrained \ac{dmt} bound converges to the original \ac{dmt}.

\begin{figure}[t]
	\centering
    \includegraphics[width=0.7\columnwidth]{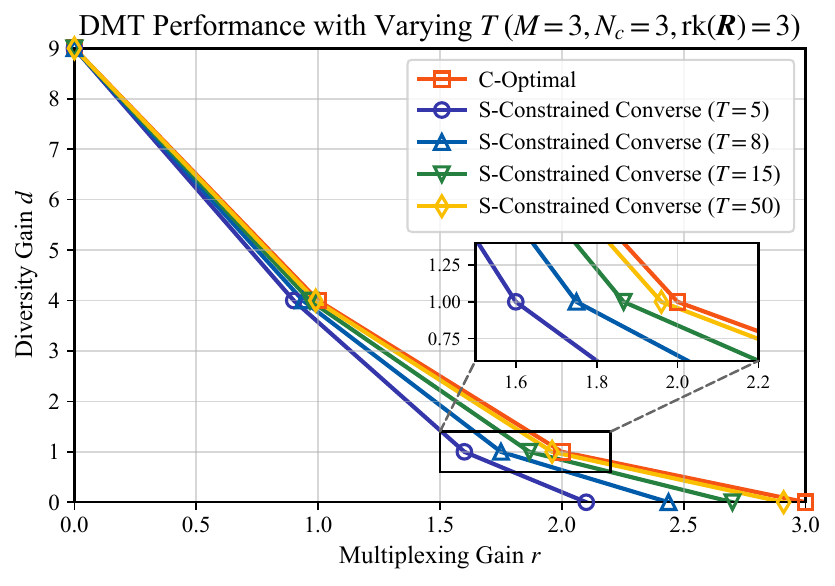}
    \vspace{-1em}
    \caption{Impact of $T$ on the sensing-constrained \ac{dmt} (see Section IV-B).}
    \vspace{-1em}
    \label{fig:DMTVaryT}
\end{figure}

\section{Conclusions}
This paper investigates the communication \ac{dmt} of a finite-blocklength \ac{mimo} \ac{isac} system constrained to sensing-optimal waveforms. Utilizing Riemannian geometry on generalized Stiefel manifolds and large-deviation analysis, we derive an elegant converse bound on the sensing-constrained \ac{dmt}. Specifically, this derived bound provides a quantitative characterization of the communication \ac{mimo} gain that is traded off to support optimal sensing.

\section{Acknowledgment}\label{sec:acknowledgment}
This research was supported in part by the National Natural Science Foundation of China under Grant 62401318, U24B20129 and 62525112.

\appendices

\section{Proof of Proposition \ref{prop:stielfelproperty}}\label{appen:prop:stiefelproperty}
\subsection{Proof of Proposition \ref{prop:stielfelproperty}-a}
\begin{proof}
Since $U(n)$ acts smoothly and transitively on $S_{\M{A}}^{(k,n)}$, $S_{\M{A}}^{(k,n)}$ is a $U(n)$-homogeneous space. 

For $\M{U}\in U(n)$, let $\M{\Delta}_1,\M{\Delta}_2\in \zz{T}_{\M{S}} S_{\M{A}}^{(k,n)}$ for arbitrary $\M{S}\in S_{\M{A}}^{(k,n)}$. 
By right-multiplication, $\M{\Delta}_1,\M{\Delta}_2$ are pushed forward to $\M{\Delta}_1\M{U},\M{\Delta}_2\M{U}\in \zz{T}_{\M{S}\M{U}} S_{\M{A}}^{(k,n)}$, since
$
    g(\M{\Delta}_1\M{U},\M{\Delta}_2\M{U})=\Re\{\text{tr}(\M{\Delta}_1\M{U} \M{U}\ermtrans \M{\Delta}_2\ermtrans)\}=g(\M{\Delta}_1,\M{\Delta}_2)
$
, $g(\cdot,\cdot)$ is $U(n)$-invariant, i.e. $\M{U}$ introduces an isometry on $S_{\M{A}}^{(k,n)}$. 

Correspondingly, $\zz{Vol}(\cdot)$ is also $U(n)$-invariant.

Finally, since for a compact homogeneous space there exists, up to a scalar, a unique group-invariant measure, the normalized measure $P_\zz{H}(\cdot)$ is the unique $U(n)$-invariant probability measure on $S_{\M{A}}^{(k,n)}$.
\end{proof}
\subsection{Proof of Proposition \ref{prop:stielfelproperty}-b}
\begin{proof}
Since there exists only one unitary-invariant probability measure on the generalized Stiefel manifold, we only need to show that the distribution of $\M{K}\RM{F}$ is $U(n)$-invariant. 

To prove this, let $\mu$ denote the probability measure of $\M{K}\RM{F}$ on $S_{\M{K}\M{A}\M{K}\ermtrans}^{(m,n)}$, for arbitrary Borel set $E\subset S_{\M{K}\M{A}\M{K}\ermtrans}^{(m,n)}$ and arbitrary $\M{U}\in U(n)$, we have:
\begin{equation}
\begin{split}
    \mu(E\M{U})&=P_\zz{H}(\M{K}^{-1}(E\M{U}))\overset{(a)}{=}P_\zz{H}(\M{K}^{-1}(E)\M{U})\\
    &\overset{(b)}{=}P_\zz{H}(\M{K}^{-1}(E))=\mu(E)
\end{split}
\end{equation}
where $(a)$ comes from the associative property of matrix multiplication and $(b)$ comes from the group-invariance of $P_\zz{H}$. Therefore, $\mu$ is $U(n)$-invariant.
\end{proof}
    
\section{Proof of Proposition \ref{prop:maxmutual}}\label{appen:prop:maxmutual}
\begin{proof}
Fix an arbitrary channel realization $\M{H}_\zz{c}$. For an arbitrary $\RM{X}\in S_{T\M{R}_0}^{(M,T)}$,
choose independent random matrix $\RM{U}\sim \mathcal{P}_{\zz{H}}$ on $U(T)$ and let 
\begin{equation}
    \RM{X}_\zz{H}\triangleq\RM{X}\RM{U},\quad \RM{Y}_\zz{H}\triangleq \sqrt{\frac{\snrc}{M}}\hcnb\RM{X}_\zz{H}+\RM{Z}_\zz{c}.
\end{equation}
By Proposition~\ref{prop:stielfelproperty}-(b), for arbitrary realization $\M{X}$, $\M{X}\RM{U}\sim\mathcal{P}_\zz{H}$ on $S_{T\M{R}_0}^{(M,T)}$, thus $\RM{X}_\zz{H}\sim \mathcal{P}_\zz{H}$ on  $S_{T\M{R}_0}^{(M,T)}$.

Since for a realization of $\RM{U}=\M{U}$,
\begin{equation}
    \RM{Y}_\zz{H}\M{U}\ermtrans=\sqrt{\frac{\snrc}{M}}\hcnb\RM{X}+\RM{Z}_\zz{c}\M{U}\ermtrans, \quad \RM{Z}_\zz{c}\overset{\zz{d}}{=}\RM{Z}_\zz{c}\M{U}
\end{equation}
thus $I_{\M{H}_\zz{c}}(\RM{X},\RM{Y}_\zz{H}|\RM{U})=I_{\M{H}_\zz{c}}(\RM{X},\RM{Y}_\zz{c})$. 
As $\RM{X}_\zz{H}$ is a deterministic function of $(\RM{X},\RM{U})$ and $(\RM{X},\RM{U})\indep \RM{Y}_\zz{H}$ conditioned on $\RM{X}_\zz{H}$, the chain rule of mutual information yields:
\begin{align}
    I_{\M{H}_\zz{c}}(\RM{X}_\zz{H};&\RM{Y}_\zz{H})=I_{\M{H}_\zz{c}}\left(\RM{X},\RM{U};\RM{Y}_\zz{H}\right)&\\&=I_{\M{H}_\zz{c}}(\RM{U};\RM{Y}_\zz{H})+I_{\M{H}_\zz{c}}(\RM{X};\RM{Y}_\zz{H}|\RM{U})
    \geq I_{\M{H}_\zz{c}}(\RM{X},\RM{Y}_\zz{c}),\nonumber
\end{align}
which proves the first equality.
Therefore, for any $\RM{X}\indep\RM{H}_\zz{c}$,
\begin{equation}
    \left\{\M{H}_\zz{c}: I_{\M{H}_\zz{c}}(\RM{X}_\zz{H};\RM{Y}_\zz{H})<TR\right\}
    \subseteq
    \left\{\M{H}_\zz{c}: I_{\M{H}_\zz{c}}(\RM{X};\RM{Y}_\zz{c})<TR\right\},
\end{equation}
proving the second equality.
\end{proof}

\section{Proof of Proposition \ref{prop:onehalf}}\label{appen:prop:onehalf}
\begin{proof}
For brevity, we use the notation
\begin{equation}
    C(\snrc,\V{\alpha}_{1:m})\triangleq I_{\V{\alpha}}\left(\sqrt{\frac{T}{M}}
        \text{diag}\left(\snrc^{0.5-\V{\alpha}_{1:m}}\right)\RM{W}_1+\RM{Z}_\zz{c1};\RM{W}_1\right).
\end{equation}
Using Proposition \ref{prop:stielfelproperty}-b, one can view $\RM{W}_1$ as the first $m$ rows of $\RM{W}$, and $\sqrt{\frac{T}{M}} \zz{diag}(\snrc^{0.5-\alpha_1},...,\snrc^{0.5-\alpha_{m}})\RM{W}_1$ as the first $m$ rows of $\RM{C}$ as they are all uniformly distributed on their corresponding generalized Stiefel manifold. 

For the sake of brevity, we use the following notations:
\begin{equation}
    \RM{W}\triangleq\left(\RM{W}_1\ermtrans,\RM{W}_2\ermtrans\right)\ermtrans,
    \RM{C}\triangleq\left(\RM{C}_1\ermtrans,\RM{C}_2\ermtrans\right)\ermtrans,
    \RM{Z}_\zz{c}\triangleq\left(\RM{Z}_\zz{c1}\ermtrans,\RM{Z}_\zz{c2}\ermtrans\right)\ermtrans
\end{equation}
where $\RM{W}_1, \RM{C}_1$ and $\RM{Z}_\zz{c1}$ denote the first $m$ rows of the corresponding matrices. Then the mutual information can be rewritten using the chain rule:
\begin{equation}
\begin{split}
    &I(\RM{C}+\RM{Z}_\zz{c};\RM{W}) = I(\RM{C}_1+\RM{Z}_\zz{c1},\RM{C}_2+\RM{Z}_\zz{c2};\RM{W}_1,\RM{W}_2) \\
    &= I(\RM{C}_1+\RM{Z}_\zz{c1};\RM{W}_1,\RM{W}_2) + I(\RM{C}_2+\RM{Z}_\zz{c2};\RM{W}_1,\RM{W}_2|\RM{C}_1+\RM{Z}_\zz{c1}) \\
    &= I(\RM{C}_1+\RM{Z}_\zz{c1};\RM{W}_1) + I(\RM{C}_1+\RM{Z}_\zz{c1};\RM{W}_2|\RM{W}_1) \\
    &\qquad\qquad\qquad + I(\RM{C}_2+\RM{Z}_\zz{c2};\RM{W}_1,\RM{W}_2|\RM{C}_1+\RM{Z}_\zz{c1}).
\end{split}
\end{equation}

Notice that given $\RM{W}_1$, $\RM{C}_1$ is completely determined. Since the noise $\RM{Z}_\zz{c1}$ is independent of $\RM{W}_2$, we have 
$I(\RM{C}_1+\RM{Z}_\zz{c1};\RM{W}_2|\RM{W}_1) = 0$. By definition, $I(\RM{C}_1+\RM{Z}_\zz{c1};\RM{W}_1) = C(\snrc,\V{\alpha}_{1:m})$. Thus, we obtain the lower bound:
\begin{equation}
\begin{split}
    I(\RM{C}&+\RM{Z}_\zz{c};\RM{W}) = C(\snrc,\V{\alpha}_{1:m}) +\\ &I(\RM{C}_2+\RM{Z}_\zz{c2};\RM{W}_1,\RM{W}_2|\RM{C}_1+\RM{Z}_\zz{c1}) 
    \geq C(\snrc,\V{\alpha}_{1:m}).
\end{split}
\end{equation}

On the other hand, expanding the conditional mutual information with differential entropy yields the upper bound:
\begin{equation}
\begin{split}
    I(\RM{C}+\RM{Z}_\zz{c};\RM{W}) &= C(\snrc,\V{\alpha}_{1:m}) + h(\RM{C}_2+\RM{Z}_\zz{c2}|\RM{C}_1+\RM{Z}_\zz{c1}) \\
    &\quad - h(\RM{C}_2+\RM{Z}_\zz{c2}|\RM{W}_1,\RM{W}_2,\RM{C}_1+\RM{Z}_\zz{c1}) \\
    &\leq C(\snrc,\V{\alpha}_{1:m}) + h(\RM{C}_2+\RM{Z}_\zz{c2}) - h(\RM{Z}_\zz{c2}).
\end{split}
\end{equation}
The inequality holds because conditioning reduces entropy, and given $\RM{W}_1, \RM{W}_2$, the signal part $\RM{C}_2$ is deterministic, leaving only the noise entropy $h(\RM{Z}_\zz{c2})$.

For $i > m$, we have $\alpha_i \geq 0.5$, which means $0.5 - \alpha_i \leq 0$. Therefore, as $\snrc \to \infty$, the power of the components in $\RM{C}_2$ is bounded (either staying constant or decaying to zero). Consequently, $h(\RM{C}_2+\RM{Z}_\zz{c2})$ is bounded, making the difference $h(\RM{C}_2+\RM{Z}_\zz{c2}) - h(\RM{Z}_\zz{c2})$ bounded by a constant.

Since $\alpha_{1:m} < 0.5$, $C(\snrc,\V{\alpha}_{1:m}) \to \infty$ as $\snrc \to \infty$. Dividing both bounds by $C(\snrc,\V{\alpha}_{1:m})$ and taking the limit, the bounded term vanishes, leading to:
\begin{equation}
    \lim_{\snrc\to\infty}\frac{C(\snrc,\V{\alpha}_{1:m})}{I(\RM{C}+\RM{Z}_\zz{c};\RM{W})} = 1.
\end{equation}
\end{proof}

\section{Proof of Theorem \ref{theorem:asyentropy}}\label{appen:theorem:asyentropy}
\begin{proof}
First we prove that $c_{\snrc,\V{\alpha}}\overset{.}{\leq}\snrc^{\alpha_m-0.5}$, which is to prove that on the generalized Stiefel manifold $S_{\snrc,\V{\alpha}}$, the three corresponding indicators in \eqref{eq:cdefinition} decline at a speed no less than $\snrc^{\alpha_m-0.5}$, respectively.

For the sake of brevity, in the following proof, we use the notation
\begin{equation}\label{eq:appen:sigmaalpha}
    \M{\Sigma}=\zz{diag}(\sigma_1,...,\sigma_m)\triangleq{\sqrt{\frac{T}{M}}{\zz{diag}(\snrc^{0.5-\V{\alpha}_{1:m}})}}
\end{equation}
where $\sigma_1\geq\sigma_2\geq...\geq\sigma_m$. Note that the definitions of $\M{\Sigma}$ and $\sigma_i$ here differ from the main text. Correspondingly, for arbitrary $\M{X}\in S_{\snrc,\V{\alpha}}$.
\begin{equation}\label{eq:appen:stiefeldefinition}
    \M{X}\M{X}\ermtrans=\M{\Sigma}^2.
\end{equation}

\subsection{Proof of $\max_{\Delta\in \zz{T}_{\M{S}}S_{\snrc,\V{\alpha}},g(\Delta,\Delta)=1}\left|\mathbf{II}_{\M{S}}(\M{\Delta},\M{\Delta})\right|\doteq\snrc^{\alpha_m-0.5}$}\label{appen:theorem:asyentropy:a}
\begin{proof}
Since $S_{\snrc,\V{\alpha}}$ is a homogeneous space, we may consider its second fundamental form at arbitrary $\M{S}\in S_{\snrc,\V{\alpha}}$. By differentiating \eqref{eq:appen:stiefeldefinition}, the tangent space T$\zz{T}_{\M{S}}S_{\snrc,\V{\alpha}}$ at $\M{S}$ is:
\begin{equation}\label{eq:appen:tangentspace}
    \zz{T}_{\M{S}}S_{\snrc,\V{\alpha}}=\left\{\M{\Delta}\in\Complex^{m\times T};\M{S}\M{\Delta}\ermtrans+\M{\Delta}\M{S}\ermtrans=\M{0}\right\}.
\end{equation}
The corresponding normal space at $\M{S}$ is:
\begin{equation}\label{eq:appen:normalspace}
    \zz{N}_{\M{S}}S_{\snrc,\V{\alpha}}=\left\{\M{\Gamma}\M{S};\M{\Gamma}=\M{\Gamma}\ermtrans\in\Complex^{m\times m}\right\}.
\end{equation}
which is easy to verify by first left-multiplying $\M{\Gamma}$ on the equation in \eqref{eq:appen:tangentspace} to get their orthogonality, and then counting two subspaces' dimensions.
Without loss of generality, set
\begin{equation}
    \M{S}=\left(\M{\Sigma},\M{0}_{m\times (T-m)}\right).
\end{equation}

To calculate $\mathbf{II}_{\M{S}}(\M{\Delta},\M{\Delta})$, we consider the unit-speed geodesic $\V{\gamma}(t)$ on $S_{\snrc,\V{\alpha}}$ starting at $\V{\gamma}(0)=\M{S}$ and initial velocity $\V{\gamma}'(0)=\M{\Delta}\in\zz{T}_{\M{S}}S_{\snrc,\V{\alpha}}$ with $g(\M{\Delta},\M{\Delta})=\zz{tr}(\M{\Delta}\M{\Delta}\ermtrans)=1$. Further differentiating \eqref{eq:appen:stiefeldefinition} along $\V{\gamma}(t)$ at $0$ yields:
\begin{equation}\label{eq:appen:secondderitive}
    \M{S}\V{\gamma}''(0)\ermtrans+\V{\gamma}''(0)\M{S}\ermtrans=-2\M{\Delta}\M{\Delta}\ermtrans.
\end{equation}
Since $\V{\gamma}(t)$ is an embedded geodesic, by basic differential geometry, we have $\V{\gamma}''(0)\perp\zz{T}_{\M{S}}S_{\snrc,\V{\alpha}}$, i.e. $\V{\gamma}''(0)\in\zz{N}_{\M{S}}S_{\snrc,\V{\alpha}}$. By \eqref{eq:appen:normalspace}, we re-express $\V{\gamma}''(0)$ as:
\begin{equation}
    \V{\gamma}''(0)=\M{\Gamma}_{\M{\Delta}}\M{S}, \quad\M{\Gamma}_{\M{\Delta}}=\M{\Gamma}_{\M{\Delta}}\ermtrans
\end{equation}
and \eqref{eq:appen:secondderitive} becomes:
\begin{equation}
    \M{S}\M{S}\ermtrans\M{\Gamma}_{\M{\Delta}}\ermtrans+\M{\Gamma}_{\M{\Delta}}\M{S}\M{S}\ermtrans=\M{\Sigma}^2\M{\Gamma}_{\M{\Delta}}\ermtrans+\M{\Gamma}_{\M{\Delta}}\M{\Sigma}^2=-2\M{\Delta}\M{\Delta}\ermtrans.
\end{equation}
Since $\M{\Sigma}$ is a diagonal matrix with positive diagonal entries, utilizing \eqref{eq:appen:stiefeldefinition}, one can directly calculate the Hermitian matrix $\M{\Gamma}_{\M{\Delta}}$ and curvature $|\V{\gamma}''(0)|$ along the geodesic as:
\begin{equation}
    \left(\M{\Gamma}_{\M{\Delta}}\right)_{ij}=-\frac{2\left(\M{\Delta}\M{\Delta}\ermtrans\right)_{ij}}{\sigma_i^2+\sigma_j^2}
\end{equation}

\begin{equation}\label{eq:appen:curineq}
\begin{split}
    &|\V{\gamma}''(0)|=\sqrt{\zz{tr}\left(\V{\gamma}''(0)\V{\gamma}''(0)\ermtrans\right)}
    =\sqrt{\zz{tr}\left(\M{\Gamma}_{\M{\Delta}}\M{\Sigma}^2\M{\Gamma}_{\M{\Delta}}\ermtrans\right)}\\
    &=\sqrt{\sum_{i,j}\sigma_i^2\left|\left(\M{\Gamma}_{\M{\Delta}}\right)_{ij}\right|^2}
    =\sqrt{\frac{1}{2}\sum_{i,j}(\sigma_i^2+\sigma_j^2)\left|\left(\M{\Gamma}_{\M{\Delta}}\right)_{ij}\right|^2}\\
    &=\sqrt{\sum_{i,j}\frac{2\left|\left(\M{\Delta}\M{\Delta}\ermtrans\right)_{ij}\right|^2}{\sigma_i^2+\sigma_j^2}}
    \overset{(a)}{\leq}\sqrt{\sum_{i,j}\frac{2\left(\M{\Delta}\M{\Delta}\ermtrans\right)_{ii}\left(\M{\Delta}\M{\Delta}\ermtrans\right)_{jj}}{\sigma_i^2+\sigma_j^2}}
\end{split}
\end{equation}
where (a) comes from noting that $\M{\Delta}\M{\Delta}\ermtrans$ is positive semi-definite. Denote the diagonal part of $\M{\Delta}\M{\Delta}\ermtrans$ by $x_i\triangleq\left(\M{\Delta}\M{\Delta}\ermtrans\right)_{ii}$, we can construct the following optimization problem for an upper bound of $|\V{\gamma}''(0)|$.

\begin{subequations}
\begin{align}
    \mathscr{P}_\zz{a-1}: \ \max\limits_{x_i}\ & {\sum_{i,j}\frac{2x_ix_j}{\sigma_i^2+\sigma_j^2}}\label{eq:appen:pa1opt}\\
    \zz{s.t.}\ \ \  
    &  \sum_i x_i=1,\ \ \ x_i\geq0. \label{eq:appen:pa1cons}
\end{align}
\end{subequations}
Note that \eqref{eq:appen:pa1opt} is a quadratic form, and using the fact that the matrix constructed by ${1}/({\sigma_i^2+\sigma_j^2})$ is positive semi-definite, the optimization function \eqref{eq:appen:pa1opt} is strictly convex. Since the feasible region \eqref{eq:appen:pa1cons} forms a simplex in $\Real^m$, thus the optimal $(x_{1,opt},...,x_{m,opt})$ lies on the simplex's corners, i.e. $x_{k,opt}=1$ for some $k\in [1,m]$. Since $\sigma_i$ is non-increasing by $i$, the solution to $\mathscr{P}_\zz{a-1}$ is $1/\sigma_m^2$ by setting $x_m=1$.

Combining the solution for $\mathscr{P}_\zz{a-1}$ and \eqref{eq:appen:curineq} yields an upper bound for $|\V{\gamma}''(0)|$, i.e.
\begin{equation}\label{eq:appen:curupb}
    |\V{\gamma}''(0)|\leq\frac{1}{\sigma_m}.
\end{equation}
Furthermore, the upper bound is tight. It can be achieved by choosing initial velocity $\V{\gamma}'(0)=\M{\Delta}_m\in  \zz{T}_{\M{S}}S_{\snrc,\V{\alpha}}$ s.t.
\begin{equation}
    \M{\Delta}_m=
    \begin{pmatrix}
    \M{0}_{(m-1)\times m} &  \M{0}_{(m-1)\times(T-m)}  \\
    \M{0}_{1\times m}        & \V{u}\\
    \end{pmatrix}
\end{equation}
with $\V{u}\V{u}\ermtrans=1$. By the Gauss formula along a curve \cite[Co. 8.3]{Lee:B18}, since $\V{\gamma}(t)$ is a unit speed geodesic, we have:
\begin{equation}\label{eq:appen:gauss}
    \mathbf{II}_{\M{S}}(\M{\Delta},\M{\Delta})=\V{\gamma}''(0).
\end{equation}
Therefore, we arrive at the desired expression:
\begin{equation}\label{eq:appen:maxsecond}
\begin{split}
       \max_{\Delta\in \zz{T}_{\M{S}}S_{\snrc,\V{\alpha}},g(\Delta,\Delta)=1} \left|\mathbf{II}_{\M{S}}(\M{\Delta},\M{\Delta})\right|&=\sqrt{\frac{M}{T}}\snrc^{\alpha_m-0.5}\\
    &\doteq\snrc^{\alpha_m-0.5}.
\end{split}
\end{equation}
\end{proof}

\subsection{Proof of $\rho_{\perp}\left(S_{\snrc,\V{\alpha}}\right)^{-1}\overset{.}{\leq}\snrc^{\alpha_m-0.5}$}
\begin{proof}
Since $S_{\M{I}_m}^{(m,T)}$ is compact, there exists a uniform tubular neighborhood around $S_{\M{I}_m}^{(m,T)}$ \cite[Thm. 5.25]{Lee:B18}. To prove the desired inequality, we show that a uniform tubular neighborhood around $S_{\snrc,\V{\alpha}}$ exists with radius $r=\sigma_m$. For arbitrary
\begin{equation}
    \M{S}=\M{\Sigma}\M{U}\in S_{\snrc,\V{\alpha}}, \M{U}\in S_{\M{I}_m}^{(m,T)}.
\end{equation}
Utilizing the normal space \eqref{eq:appen:normalspace}, we construct the Fermi coordinate representation of the tubular neighborhood by choosing $\M{\Gamma}=\M{\Gamma}\ermtrans$ and let
\begin{equation}\label{eq:appen:fermi}
    \M{X}=\M{S}+\M{\Gamma}\M{S}=(\M{I}+\M{\Gamma})\M{\Sigma}\M{U}
\end{equation}
and the uniform tubular neighborhood $T(S_{\snrc,\V{\alpha}},r)$ around $S_{\snrc,\V{\alpha}}$ of radius $r$ is:
\begin{equation}\label{eq:appen:tubularnbhd}
\begin{split}
    &T(S_{\snrc,\V{\alpha}},r)=\\
    &\left\{(\M{I}+\M{\Gamma})\M{\Sigma}\M{U};\M{U}\M{U}\ermtrans=\M{I}_m, \M{\Gamma}=\M{\Gamma}\ermtrans,
    \|\M{\Gamma}\M{\Sigma}\M{U}\|_\zz{F} < r
    \right\}
\end{split}
\end{equation}
To show that $T(S_{\snrc,\V{\alpha}},\sigma_m)$ is indeed a tubular neighborhood, using the theorem of invariance of domain one needs to prove \eqref{eq:appen:fermi} is an injection. Consider the Euclidean projection of $\M{X}\in T(S_{\snrc,\V{\alpha}},r)$ to $S_{\snrc,\V{\alpha}}$, i.e. the following optimization problem
\begin{subequations}
\begin{align}
    \mathscr{P}_\zz{a-2}: \ \min\limits_{\tilde{\M{U}}}\ & \|(\M{I}+\M{\Gamma})\M{\Sigma}\M{U}-\M{\Sigma}\tilde{\M{U}}\|_\zz{F}\label{eq:appen:pa2opt}\\
    \zz{s.t.}\ \ \  
    &  \tilde{\M{U}}\tilde{\M{U}}\ermtrans=\M{I}_m. \label{eq:appen:pa2cons}
\end{align}
Minimizing \eqref{eq:appen:pa2opt} is equivalent to maximizing
\begin{align}
    \max\limits_{\tilde{\M{U}}}\  &\Re\left\{\zz{tr}\left[(\M{I}+\M{\Gamma})\M{\Sigma}\M{U}\tilde{\M{U}}\ermtrans\M{\Sigma}\right]\right\}\\
    =&\Re\left\{\zz{tr}\left[\M{\Sigma}(\M{I}+\M{\Gamma})\M{\Sigma}\M{U}\tilde{\M{U}}\ermtrans\right]\right\}.\label{eq:appen:optimizemid}
\end{align}
\end{subequations}
Since $\|\M{\Gamma}\M{\Sigma}\M{U}\|_\zz{F}\leq \sigma_m$, using the fact $0\leq\sigma_m\leq\sigma_i$, we have
\begin{equation}
    \sigma_m^2\zz{tr}\left(\M{\Gamma}\M{\Gamma}\ermtrans\right)\leq\zz{tr}\left(\M{\Gamma}\M{\Sigma}^2\M{\Gamma}\ermtrans\right)<\sigma_m^2
\end{equation}
Therefore, any eigenvalue $\lambda(\M{\Gamma})$ of $\M{\Gamma}$,
\begin{equation}
    |\lambda(\M{\Gamma})|< 1,\quad \M{\Sigma}(\M{I}+\M{\Gamma})\M{\Sigma}\succ \M{\Sigma}(\M{I}-\M{I})\M{\Sigma}=\M{0}
\end{equation}
and let $\M{\Sigma}(\M{I}+\M{\Gamma})\M{\Sigma}=\M{U}_1\M{\Lambda}_1\M{U}_1\ermtrans$ be the eigenvalue decomposition, the optimization objective \eqref{eq:appen:optimizemid} further becomes:
\begin{equation}
    \max\limits_{\tilde{\M{U}}}\  \Re\left\{\zz{tr}\left(\M{\Lambda}_1\M{U}_1\ermtrans\M{U}\tilde{\M{U}}\ermtrans\M{U}_1\right)\right\}.
\end{equation}
Since $\M{\Lambda}_1\succ \M{0}$ and $\M{U}_1\ermtrans\M{U}\tilde{\M{U}}\ermtrans\M{U}_1$ is also unitary, it is trivial that the only solution to $\mathscr{P}_\zz{a-2}$ is
\begin{equation}\label{eq:appen:opt2sol}
    \M{U}_1\ermtrans\M{U}\tilde{\M{U}}\ermtrans\M{U}_1=\M{I}, \quad\text{i.e.}\quad\tilde{\M{U}}=\M{U}.
\end{equation}
That is, arbitrary matrix $\M{X}=\M{S}+\M{\Gamma}\M{S}\in T(S_{\snrc,\V{\alpha}},r)$ has unique Euclidean projection $\M{S}$ w.r.t. $S_{\snrc,\V{\alpha}}$. If $\M{X}$ can also be expressed as $\M{X}=\M{S}_1+\M{\Gamma}_1\M{S}_1$, the uniqueness of the Euclidean projection doesn't hold. Therefore $T(S_{\snrc,\V{\alpha}},\sigma_m)$ is a uniform tubular neighborhood indeed. And we arrive at:
\begin{equation}
    \rho_{\perp}\left(S_{\snrc,\V{\alpha}}\right)^{-1}\leq \frac{1}{\sigma_m}=\sqrt{\frac{M}{T}}\snrc^{\alpha_m-0.5}\doteq \snrc^{\alpha_m-0.5}.
\end{equation}
\end{proof}

\subsection{Proof of $\rho_{\top}\left(S_{\snrc,\V{\alpha}}\right)^{-1}\overset{.}{\leq}\snrc^{\alpha_m-0.5}$}
\begin{proof}
To prove the desired inequality, we use the following celebrated Klingenberg's Theorem:
\begin{lemma}[Klingenberg]\label{lemma:appen:Klin}
{Let $(M,g)$ be a compact Riemannian manifold whose sectional curvature satisfies $K \le C$ for some constant $C$. Then either}
\begin{equation}
    \rho_\top(M) \ge \frac{\pi}{\sqrt{C}}
\end{equation}
{or there exists a closed geodesic $\gamma$ in $M$ whose length is minimum among all closed geodesics, such that}
\begin{equation}
    \rho_\top(M) = \frac{1}{2}L(\gamma).
\end{equation}
\end{lemma}
Since $S_{\snrc,\V{\alpha}}$ is compact, Lemma \ref{lemma:appen:Klin} is suitable, thus we focus on the upper bound of the sectional curvature $K$ and the minimum distance $L(\gamma)$ of a closed geodesic $\gamma$ on $S_{\snrc,\V{\alpha}}$.

First, note that the sectional curvature $K$ is bounded by:
\begin{equation}\label{eq:appen:scurvatureupbound}
   K=g(\mathbf{II}_{\M{S}}(\M{\Delta}_1,\M{\Delta}_1),\mathbf{II}_{\M{S}}(\M{\Delta}_2,\M{\Delta}_2))\leq \max_{\Delta} \left|\mathbf{II}_{\M{S}}(\M{\Delta},\M{\Delta})\right|^2
\end{equation}
where unit tangent vectors $\M{\Delta}_1,\M{\Delta}_2,\M{\Delta}\in\zz{T}_{\M{S}}S_{\snrc,\V{\alpha}}$ with $\M{\Delta}_1\perp\M{\Delta}_2$. 

Second, for arbitrary closed geodesic $\V{\gamma}(t)$ in $S_{\snrc,\V{\alpha}}$, according to the analysis in Section \ref{appen:theorem:asyentropy:a}, we have:
\begin{equation}
    |\V{\gamma}''(t)|=|\mathbf{II}_{\V{\gamma}(t)}(\V{\gamma}'(t),\V{\gamma}'(t))|\leq\frac{1}{\sigma_m}
\end{equation}
therefore, using Fenchel theorem, for $\gamma$ to be closed, we have
\begin{equation}\label{eq:appen:lengthup}
    L(\V{\gamma})\geq {2\pi}{\sigma_m}.
\end{equation}
Therefore, combining \eqref{eq:appen:maxsecond}, \eqref{eq:appen:scurvatureupbound} and \eqref{eq:appen:lengthup}, Lemma \ref{lemma:appen:Klin} yields our desired result
\begin{equation}
    \rho_{\top}\left(S_{\snrc,\V{\alpha}}\right)^{-1}\overset{.}{\leq}\snrc^{\alpha_m-0.5}.
\end{equation}
\end{proof}

Combining these results, we arrive at:
\begin{equation}
    c_{\snrc,\V{\alpha}}\overset{.}{\leq}\snrc^{\alpha_m-0.5}
\end{equation}
substituting the inequality to \eqref{eq:boundedentropy} and focusing on its RHS, if $\alpha_m\leq 0$, the RHS obviously converges to zero as $\snrc\to\infty$, if $\alpha_m>0$, we have:
\begin{equation}
\begin{split}
    \quad&\lim_{\snrc\to\infty}\delta^{-1}(1+c_{\snrc,\V{\alpha}}\snrc^{0.5})^{\delta}c_{\snrc,\V{\alpha}}^2\log^2(c_{\snrc,\V{\alpha}})\\
    =&\lim_{\snrc\to\infty}\delta^{-1}(1+\snrc^{\alpha_m})^{\delta}\snrc^{2\alpha_m-1}\log^2(\snrc^{2\alpha_m-1})\\
    =&\lim_{\snrc\to\infty}{(\delta(2\alpha_m-1))}^{-1}\snrc^{(2+\delta)\alpha_m-1}\log^2(\snrc)\\
    \overset{(a)}{=}&0
\end{split}
\end{equation}
where (a) comes from choosing $0<\delta<1/\alpha_m-2$. Thus by inequality \eqref{eq:boundedentropy}, we have
\begin{equation}\label{eq:appen:entropylimit}
    \lim_{\snrc\to\infty}\left|h\left(\RM{C}_1+\RM{Z}_\zz{c1}\right)-\frac{m^2}{2}\log(\pi e)-\log \zz{Vol}\left(S_{\snrc,\V{\alpha}}\right)\right|=0.
\end{equation}
Next, we show that
\begin{equation}\label{eq:appen:volume}
    \begin{split}
    \zz{Vol}&\left(S_{\snrc,\V{\alpha}}\right)=\\
    &2^{-\frac{m(m-1)}{2}}\prod_{i<j}(\sigma_i^2+\sigma_j^2)\prod_i\sigma_i^{2T-2m+1}\zz{Vol}(S_{\M{I}_m}^{(m,T)})
    \end{split}
\end{equation}
where
\begin{equation}
    \zz{Vol}(S_{\M{I}_m}^{(m,T)})=\prod_{i=T-m+1}^{T}\frac{2\pi^i}{(i-1)!}
\end{equation}
is the Euclidean volume of a standard complex Stiefel manifold \cite{ZheTse:J02}.
\subsection{Proof of \eqref{eq:appen:volume}}
\begin{proof}
Consider the diffeomorphism 
\begin{equation}
    \phi: S_{\M{I}_m}^{(m,T)}\to S_{\snrc,\V{\alpha}}, \quad\M{X}\mapsto\M{\Sigma}\M{X}
\end{equation}
Let $dV_1$ and $dV_2$ denote the Euclidean Riemannian volume form on $S_{\M{I}_m}^{(m,T)}$ and $S_{\snrc,\V{\alpha}}$, respectively. By definition:
\begin{equation}\label{eq:appen:volumetrans}
    \zz{Vol}(S_{\snrc,\V{\alpha}})=\int_{S_{\snrc,\V{\alpha}}}dV_2=\int_{S_{\M{I}_m}^{(m,T)}}\phi^* dV_2=\int_{S_{\M{I}_m}^{(m,T)}}f dV_1
\end{equation}
for some function $f$ on $S_{\M{I}_m}^{(m,T)}$, where $\phi^*$ denotes the pull-back operation of differential forms. Since both manifolds are $U(T)$-homogeneous, and utilizing the associative property of matrix multiplication, for arbitrary $\M{U}\in U(T)$, let $u: \M{X}\mapsto \M{X}\M{U}$, for arbitrary $\M{S}\in S_{\M{I}_m}^{(m,T)}$, equation \eqref{eq:appen:invariance} holds, thus $f$ is a constant on $S_{\M{I}_m}^{(m,T)}$.

\begin{figure*}[th]\label{eq:appen:invariance}
\begin{equation}
    f_{u(\M{S})}dV_1|_{u(\M{S})}=\phi^*dV_2|_{\phi\circ u(\M{S})}=\phi^*dV_2|_{u\circ\phi(\M{S})}=(u\circ u^{-1}\circ\phi)^*dV_2|_{u\circ\phi(\M{S})}
    =( u^{-1}\circ\phi)^*dV_2|_{\phi(\M{S})}=(u^{-1})^*f_{\M{S}}dV_1|_{\M{S}}=f_{\M{S}}dV_1|_{u(\M{S})}
\end{equation}
\begin{equation}\label{eq:appen:fcalculation}
    f=fdV_1(\M{V}_{1:5})=\phi^*dV_2(\M{V}_{1:5})=dV_2(\phi_*(\M{V}_{1:5}))=dV_2\big(\sqrt{\frac{{\sigma_i^2+\sigma_j^2}}{2}}\M{W}_{1,2,ij},\sigma_i\M{W}_{3:5,ij}\big)=2^{-\frac{m(m-1)}{2}}\prod_{i<j}(\sigma_i^2+\sigma_j^2)\prod_i\sigma_i^{2T-2m+1}
\end{equation}
\begin{equation}\label{eq:appen:entropyform}
    \lim_{\snrc\to\infty}\left|h\left(\RM{C}_1+\RM{Z}_\zz{c1}\right)-\frac{m^2}{2}\log(\pi e)-\log \left(2^{-\frac{m(m-1)}{2}}\prod_{i<j}(\sigma_i^2+\sigma_j^2)\prod_i\sigma_i^{2T-2m+1}\prod_{i=T-m+1}^{T}\frac{2\pi^i}{(i-1)!})\right)\right|=0
\end{equation}
\begin{equation}\label{eq:appen:mutualinformationform}
    \lim_{\snrc\to\infty}\left|I_{\RM{H}_\zz{c}}({\RM{X}};\RM{Y}_\zz{c})+\frac{m(2T-m)}{2}\log(\pi e)-\log \left(2^{-\frac{m(m-1)}{2}}\prod_{i<j}(\sigma_i^2+\sigma_j^2)\prod_i\sigma_i^{2T-2m+1}\prod_{i=T-m+1}^{T}\frac{2\pi^i}{(i-1)!})\right)\right|=0
\end{equation}
\rule{\linewidth}{0.5pt}
\end{figure*}

Without loss of generality, we consider $f$ at
\begin{equation}
    \M{S}=\left(\M{I}_m,\M{0}_{m\times (T-m)}\right), \phi(\M{S})=\left(\M{\Sigma},\M{0}_{m\times (T-m)}\right)
\end{equation}
and find the orthonormal basis for $\zz{T}_{\M{S}}S_{\M{I}_m}^{(m,T)}$ and $\zz{T}_{\phi(\M{S})}S_{\snrc,\V{\alpha}}$. Specifically, the following tangent matrices form an orthonormal basis for $\zz{T}_{\M{S}}S_{\M{I}_m}^{(m,T)}$:
\begin{subequations}
\begin{align}
    \M{V}_{1,ij}&=\left(\frac{1}{\sqrt{2}}(\M{E}_{ij}-\M{E}_{ji}),\M{0}_{m\times (T-m)}\right),\quad i<j\\
    \M{V}_{2,ij}&=\sqrt{-1}\left(\frac{1}{\sqrt{2}}(\M{E}_{ij}+\M{E}_{ji}),\M{0}_{m\times (T-m)}\right),\quad i<j\\
    \M{V}_{3,i}&=\sqrt{-1}\left(\M{E}_{ii},\M{0}_{m\times (T-m)}\right)\\
    \M{V}_{4,ij}&=\left(\M{0}_{m\times m},\M{E}_{ij}\right)\\
    \M{V}_{5,ij}&=\sqrt{-1}\left(\M{0}_{m\times m},\M{E}_{ij}\right)\\
\end{align}
\end{subequations}
and we organize them as $\M{V}_{1:5}$. Similarly, a set of orthonormal basis for $\zz{T}_{\phi(\M{S})}S_{\snrc,\V{\alpha}}$ is:
\begin{subequations}
\begin{align}
    \M{W}_{1,ij}&=\left(\frac{\sigma_i\M{E}_{ij}-\sigma_j\M{E}_{ji}}{\sqrt{\sigma_i^2+\sigma_j^2}},\M{0}_{m\times (T-m)}\right),\quad i<j\\
    \M{W}_{2,ij}&=\sqrt{-1}\left(\frac{\sigma_i\M{E}_{ij}+\sigma_j\M{E}_{ji}}{\sqrt{\sigma_i^2+\sigma_j^2}},\M{0}_{m\times (T-m)}\right),\quad i<j\\
    \M{W}_{3,i}&=\sqrt{-1}\left(\M{E}_{ii},\M{0}_{m\times (T-m)}\right)\\
    \M{W}_{4,ij}&=\left(\M{0}_{m\times m},\M{E}_{ij}\right)\\
    \M{W}_{5,ij}&=\sqrt{-1}\left(\M{0}_{m\times m},\M{E}_{ij}\right)\\
\end{align}
\end{subequations}
Therefore we have:
\begin{subequations}
\begin{align}
    \phi_*(\M{V}_{1,2,ij})=\M{\Sigma}\M{V}_{1,2,ij}&=\sqrt{\frac{{\sigma_i^2+\sigma_j^2}}{2}}\M{W}_{1,2,ij}\\
    \phi_*\M{V}_{3:5,ij}=\M{\Sigma}\M{V}_{3:5,ij}&=\sigma_i\M{W}_{3:5,ij}
\end{align}
\end{subequations}
where $\phi_*$ denotes the push-forward of tangent vectors, thus $f$ can be calculated by \eqref{eq:appen:fcalculation}. The desired \eqref{eq:appen:volume} then comes from using \eqref{eq:appen:volumetrans}.
\end{proof}

Substituting our results into \eqref{eq:appen:sigmaalpha} yields the asymptotic differential entropy and mutual information in \eqref{eq:appen:entropyform} and \eqref{eq:appen:mutualinformationform}.

Utilizing \eqref{eq:appen:sigmaalpha} and the fact that $\alpha_i\leq\alpha_{i+1}$, for $i<j$, we have
\begin{equation}
    \left(\sigma_i^2+\sigma_j^2\right)\doteq\snrc^{1-2\alpha_i}
\end{equation}
and the asymptotic entropy and mutual information is thus:
\begin{align}
    I_{\hcnb}({\RM{X}};\RM{Y}_\zz{c}) &\sim h_{\V{\alpha}}(\RM{C}_1+\RM{Z}_\zz{c1}) \nonumber \\
    &\sim \sum_{i=1}^{\rank{R}}(2T+1-2i)(0.5-\alpha_i)^+\log(\snrc).
\end{align}
\end{proof}

\section{Proof of Lemma \ref{lemma:singularvalue}}\label{appen:lemma:singularvalue}
\subsection{Proof of \eqref{eq:alphadistribution}}
\begin{proof}
Let the eigendecomposition of the transmit covariance be $\M{R}=\M{U}\M{\Lambda}\M{U}\ermtrans$, with $\M{\Lambda}=\zz{diag}(\lambda_1,\lambda_2,...,\lambda_{\rank{R}},0,...,0)\triangleq \zz{diag}(\V{\lambda}\trans,0,...,0)$,
Since elements of $\hc$ are i.i.d. white complex Gaussian distributed and $\RM{U}\in U(M)$, let $\M{\Lambda}_\zz{>0}=\zz{diag}(\V{\lambda})$, we have:
\begin{equation}
    \hc\M{R}\hc\ermtrans=\hc\M{U}\M{\Lambda}\M{U}\ermtrans\hc\ermtrans\sim\hc\M{\Lambda}\hc\ermtrans=\RM{H}_\zz{c1}\M{\Lambda}_\zz{>0}\RM{H}_\zz{c1}\ermtrans
\end{equation}
where $\RM{H}_\zz{c1}$ denotes the first $\rank{R}$ columns of $\hc$. Since $\RM{H}_\zz{c1}\M{\Lambda}_\zz{>0}\RM{H}_\zz{c1}\ermtrans=(\M{\Lambda}_{>0}^{1/2}\RM{H}_\zz{c1}\ermtrans)\ermtrans(\M{\Lambda}_{>0}^{1/2}\RM{H}_\zz{c1}\ermtrans)$ has the same eigenvalues as $(\M{\Lambda}_{>0}^{1/2}\RM{H}_\zz{c1}\ermtrans)(\M{\Lambda}_{>0}^{1/2}\RM{H}_\zz{c1}\ermtrans)\ermtrans$, one only needs to consider the eigenvalue distribution of 
\begin{equation}
    (\M{\Lambda}_{>0}^{1/2}\RM{H}_\zz{c1}\ermtrans)(\M{\Lambda}_{>0}^{1/2}\RM{H}_\zz{c1}\ermtrans)\ermtrans\sim \mathcal{W}_{\rank{R}}(\nc,\M{\Lambda}_{>0})
\end{equation}
where the RHS stands for the central complex Wishart distribution with $\nc$ degrees of freedom and covariance matrix $\M{\Lambda}_{>0}$. The eigenvalue distribution of such a random matrix is\cite{Jam:J64}:
\begin{equation}
    \frac{\det(\{e^{-a_i/\lambda_j}\})}{\det\M{\Lambda}_{>0}^{\nc}} \prod_{\ell=1}^{\rank{R}} \frac{a_\ell^{\nc-\rank{R}}}{(\nc-\ell)!} \prod_{k<\ell}^{\rank{R}} \frac{a_k - a_\ell}{\lambda_k - \lambda_\ell} \lambda_\ell \lambda_k.
\end{equation}
where $a_i=\sigma_i^2$ are the eigenvalues. We arrive at \eqref{eq:alphadistribution} by utilizing change of variables:
\begin{equation}
    \alpha_i=-\frac{\log a_i}{2\log \snrc}
\end{equation}
and accounting for the Jacobian determinant.
\end{proof}

\subsection{Proof of \eqref{eq:alphadistributionasymp}}
\begin{proof}
First, we focus on the part $ \alpha_{1}>0$, i.e. $\alpha_i>0$ for arbitrary $i\in \{1,2,...,\rank{R}\}$. Again, let $x_i=\snrc^{-2\alpha_i}$, as $\snrc\to\infty$, $x_i\to 0^+$,
we have
\begin{equation}\label{eq:appen:palpha}
    p(\V{\alpha})\propto\det\left(\left\{e^{-\frac{x_i}{\lambda_j}}\right\}\right)
    \prod_{m=1}^{\rank{R}}{x_m}^{(\nc-\rank{R}+1)}\prod_{k<l}^{\rank{R}}\left(x_k-x_l\right)
\end{equation}
focusing on the first term $\det\left(\left\{e^{-\frac{x_i}{\lambda_j}}\right\}\right)$ on the RHS, as $x_i\to 0^+$, we have
\begin{equation}
\begin{split}
    \det&\left(\left\{e^{-\frac{x_i}{\lambda_j}}\right\}\right)\\
    &=\det\left(\left\{\sum_{k=0}^{\rank{R}-1}\left(\frac{(-1)^k}{\lambda_j^k}\frac{1}{k!}x_i^k\right)(1+o(1))\right\}\right)\\
    &=\det\left(\left\{\sum_{k=0}^{\rank{R}-1}\frac{(-1)^k}{\lambda_j^k}\frac{1}{k!}\right\}\right)\det\left(\left\{x_i^k\right\}\right)(1+o(1))\\
    &\propto\prod_{k<l}^{\rank{R}}\left(x_k-x_l\right)(1+o(1)).
\end{split}
\end{equation}
Therefore,
\begin{equation}
\begin{split}
    p(\V{\alpha})&\propto
    \prod_{m=1}^{\rank{R}}{x_m}^{(\nc-\rank{R}+1)}\prod_{k<l}^{\rank{R}}\left(x_k-x_l\right)^2\\
    &\doteq  {\snrc}^{-\sum_{i=1}^{\rank{R}}2(\nc+\rank{R}+1-2i)\alpha_i}.
\end{split}
\end{equation}
and the first part of \eqref{eq:alphadistributionasymp} is proved. 

For the other part, if $\alpha_1<0$, by the decomposition law for determinants, the determinant term $\det\left(\left\{e^{\frac{-\snrc^{-2\alpha_i}}{\lambda_j}}\right\}\right)$ in \eqref{eq:appen:palpha} becomes
\begin{equation}
    \sum_k^{\rank{R}}(-1)^{k+1} e^{\frac{-\snrc^{-2\alpha_1}}{\lambda_k}}
    \underset{i\neq 1,j\neq k}{\det}\left(\left\{e^{\frac{-\snrc^{-2\alpha_i}}{\lambda_j}}\right\}\right).
\end{equation}
It is then trivial that $p(\V{\alpha})$ scales with $\snrc$ faster than $e^{-\snrc^{-2\alpha_1}}$, thus yielding the desired
\begin{equation}
    \lim_{\snrc\to\infty}\frac{\log p_{\RV{\alpha}}(\V{\alpha})}{\log\snrc}=-\infty.
\end{equation}
\end{proof}

\section{Proof of Theorem \ref{theo:final}}\label{appendix:theo:final}
\begin{proof}
To find $d^{\zz{out}}_{\rank{R}}(r)$ is to find the solution for the following optimization problem:
\begin{subequations}
\begin{align}
    \mathscr{P}_\zz{a-3}: \ \min\limits_{\V{\alpha}}\ & 2\left(\sum_{i=1}^{\rank{R}}(\nc+\rank{R}+1-2i)\alpha_i\right)\label{eq:appen:pa3opt}\\
    \zz{s.t.}\ \ \  
    &  \sum_{i}^{\rank{R}}(2T+1-2i)(0.5-\alpha_i)^+\leq Tr\\ \label{eq:appen:pa3cons}
    & 0\leq\alpha_1\leq\alpha_2\leq...\leq\alpha_{\min\left\{\nc, \rank{R}\right\}}
\end{align}
\end{subequations}
where \eqref{eq:appen:pa3opt} is $d^{\zz{out}}_{\rank{R}}(r)$. Similar to the approach in \cite{ZheTse:J03}, notice that the optimal solution for $r=k(1-\frac{k}{2T})\triangleq r(k)$ is
\begin{equation}
    \alpha_i=\begin{cases}
    0, & 1\leq i\leq k\\
    0.5, & k<i\leq\rank{R}.
    \end{cases}
\end{equation}
with $d^{\zz{out}}_{\rank{R}}(r(k))$ being:
\begin{equation}
    \begin{split}
    d^{\zz{out}}_{\rank{R}}(r(k))
    &=\sum_{i=k+1}^{\rank{R}}(\nc+\rank{R}+1-2i)\\
    &=(\nc-k)(\rank{R}-k).
    \end{split}
\end{equation}
For other $r$, as $r$ grows from $r(k)$ to $r(k)<r<r(k+1)$, the optimal $\alpha_{k+1}$ falls from $0.5$ towards $0$. It follows easily that the corresponding $d^{\zz{out}}_{\rank{R}}(r)$ is the affine combination of $d^{\zz{out}}_{\rank{R}}(k)$ and $d^{\zz{out}}_{\rank{R}}(k+1)$  as \eqref{eq:appen:pa3opt} is linear in $\V{\alpha}$.

After obtaining the outage probability, the outage bound \eqref{eq:newdmtupbound} can be proved using the same proof as \cite[Lem. 5]{ZheTse:J03} thus we omit the derivation for brevity.
\end{proof}

\ifCLASSOPTIONcaptionsoff
  \newpage
\fi
 \newpage
\end{document}